\documentclass[lettersize,journal]{IEEEtran}
\usepackage{amsmath,amsfonts}
\usepackage{array}
\usepackage{textcomp}
\usepackage{stfloats}
\usepackage{verbatim}
\usepackage{graphicx}
\usepackage{cite}

\usepackage{url}
\usepackage{algorithmicx}
\usepackage{algorithm}
\usepackage{float}
\usepackage{algpseudocode}
\usepackage{booktabs}
\usepackage{multirow}
\usepackage{float}
\usepackage{amsmath, amssymb, amsthm}  
\newtheorem{theorem}{Theorem}

\usepackage{lineno} 
\usepackage{subcaption}
\usepackage{overpic}
\usepackage[table]{xcolor}  
\usepackage{colortbl}
\usepackage{wrapfig}
\usepackage{booktabs}
\usepackage{lineno} 

\usepackage[normalem]{ulem} 

\newtheorem{lemma}{Lemma}

\newtheorem{property}{Property}

\newtheorem{definition}{Definition}
\begin{document}

\title{Weighted Squared Volume Minimization (WSVM) for Generating Uniform Tetrahedral Meshes}

\author{Kaixin~Yu$^*$,
        Yifu~Wang$^*$,
        Peng~Song,
        Xiangqiao~Meng,
        Ying~He$^\dagger$,
        Jianjun~Chen$^\dagger$
\thanks{$^*$K. Yu and Y. Wang contributed equally to the paper.}
\thanks{$^\dagger$Corresponding authors: Y. He (yhe@ntu.edu.sg) and J. Chen (jjchen@zju.edu.cn).}
\thanks{Kaixin Yu, Yifu Wang and Jianjun Chen are with the School of Aeronautics and Astronautics, Zhejiang University, Hangzhou, China. Kaixin Yu is also affiliated with the College of Computing and Data Science, Nanyang Technological University, Singapore.}
\thanks{Peng Song is with the Pillar of Information Systems Technology and Design, Singapore University of Technology and Design, Singapore.}
\thanks{Xiangqiao Meng is with the Department of Computing, The Hong Kong Polytechnic University, Hong Kong, China.}
\thanks{Ying He is with the College of Computing and Data Science, Nanyang Technological University, Singapore.}
}

\maketitle

\begin{abstract}
This paper presents a new algorithm, Weighted Squared Volume Minimization (WSVM), for generating high-quality tetrahedral meshes from closed triangle meshes. Drawing inspiration from the principle of minimal surfaces that minimize squared surface area, WSVM employs a new energy function integrating weighted squared volumes for tetrahedral elements. When minimized with constant weights, this energy promotes uniform volumes among the tetrahedra. Adjusting the weights to account for local geometry further achieves uniform dihedral angles within the mesh. The algorithm begins with an initial tetrahedral mesh generated via Delaunay tetrahedralization and proceeds by sequentially minimizing volume-oriented and then dihedral angle-oriented energies. At each stage, it alternates between optimizing vertex positions and refining mesh connectivity through the iterative process. The algorithm operates fully automatically and requires no parameter tuning. Evaluations on a variety of 3D models demonstrate that WSVM consistently produces tetrahedral meshes of higher quality, with fewer slivers and enhanced uniformity compared to existing methods. Check out further details at the project webpage: \url{https://kaixinyu-hub.github.io/WSVM.github.io/}.
\end{abstract}


\section{Introduction}
\label{sec:introduction}

\IEEEPARstart{T}{etrahedral} meshes provide a natural representation of solids. Compared to other volumetric representations, such as voxels and implicit functions, they are valued for their superior adaptability and capacity to represent complex geometries and topologies. The generation of high-quality tetrahedral meshes is a fundamental problem in digital geometry processing \cite{hang2015tetgen,cheng2013delaunay}, with applications spanning computer graphics \cite{dey2004approximating,cabiddu2017,hu2018tetwild,hu2020fasttetwild,gao2020learning,csahistan2021ray,plocharski2024skeleton}, scientific computations \cite{baker1989automatic,owen1998survey,schneider2022large} and various engineering domains \cite{biswas1998tetrahedral,yu2008high,zhang2021meshingnet3d}.  

The last three decades have seen significant advancements in tetrahedral mesh generation, with numerous efficient and practical algorithms developed. These include advancing-front methods~\cite{lohner1988generation,jin1993AFTgeneration}, envelope methods~\cite{hu2018tetwild,hu2019triwild,hu2020fasttetwild}, Laplacian methods~\cite{field1988laplacian,freitag1995efficient,freitag1996comparison,vollmer1999improvedLaplace,sorkine2006differential, klingner2009tetrahedral}, and a broad array of techniques related to Delaunay tetrahedralizations~\cite{caendish1985apporach,shewchuk2002good,sheehy2012new,hang2015tetgen,sorgente2023survey} and Voronoi diagrams~\cite{du1999centroidal,du2005anisotropic,liu2009centroidal,levy2010p,yan2013efficient,ekelschot2018robust,xiao2018optimal,ekelschot2019parallel,telsang2022computation}. Despite these advancements, generating high-quality tetrahedral meshes remains a significant challenge, as completely eliminating slivers is both theoretically and practically difficult \cite{fu2014anisotropic}. The presence of slivers can adversely impact the subsequent applications of tetrahedral meshes. 

This paper presents an approach to enhance the quality of tetrahedral meshes by reducing slivers and improving uniformity. Inspired by minimal surfaces that minimize squared surface area~\cite{renka1995minimal,twosimple,renka2015mesh}, 
our method introduces a new energy function that integrates weighted squared volumes for tetrahedral elements. When minimized with constant weights, this function promotes uniform volumes across tetrahedra. Adjusting the weights to account for local geometry results in uniform dihedral angles within the mesh. 
A key feature of our energy is its convexity and differentiability, which permits using Newton's solver for efficient minimization. 

Building on this foundation, we develop an iterative algorithm named Weighted Squared Volume Minimization (WSVM). Starting with an initial tetrahedral mesh generated by Delaunay tetrahedralization, WSVM sequentially optimizes the position of each interior vertex by first minimizing the volume-oriented energy and then angle-oriented energy. The algorithm also dynamically refines mesh connectivity by splitting edges and flipping faces and edges throughout the iterative process. Our algorithm is fully automatic and operates without the need for parameter tuning. We test WSVM on 100 models, each featuring diverse geometric and topological characteristics, to evaluate its performance against state-of-the-art (SOTA) methods. Experimental results show that WSVM consistently produces tetrahedral meshes of higher quality, with fewer slivers and enhanced uniformity compared with existing methods. 


\section{Related Work}
\label{sec:relatedwork}
Existing tetrahedralization methods are classified into two main categories: approaches that generate tetrahedral meshes from input boundary meshes and optimization-based techniques that enhance mesh quality.


\subsection{
Generation Techniques}

\paragraph{Delaunay methods} Delaunay tetrahedralization, due to its efficiency and favorable mathematical properties, is widely used in commercial software. Over the past two decades, research on Delaunay tetrahedralization algorithms for point sets has seen significant advancements \cite{sheehy2012new}. Despite these advances, constrained Delaunay tetrahedralization remains challenging. A notable issue is the need for Steiner points to facilitate tetrahedralization~\cite{hang2015tetgen}. Additionally, unlike 2D Delaunay triangulation, where the empty circumcircle property helps maximize the minimum angle, 3D Delaunay tetrahedral meshes frequently contain slivers—elements that satisfy Delaunay conditions yet are of poor quality \cite{caendish1985apporach, talmor1997well}. These slivers can significantly degrade the accuracy of numerical simulations \cite{shewchuk2002good, sorgente2023survey}. 

\paragraph{Advancing-front methods} 
The advancing-front method, introduced by \cite{lohner1988generation},
creates an unstructured mesh by sequentially adding individual elements to an existing front of generated elements. The process begins with generating boundary nodes, whose edges form the initial front that advances outward into the field. This method and its many variants are widely used for generating planar, surface, and tetrahedral meshes \cite{jin1993AFTgeneration}. However, they rely on heuristic strategies to generate a locally optimal element at each step, employing a flood-filling manner to fill the space. These heuristic approaches do not guarantee convergence and can sometimes lead to issues with robustness. While the advancing-front method can produce high-quality elements at boundaries, the mesh quality at internal intersections of advancing fronts can be suboptimal. 

\paragraph{Envelope methods} These methods stand out for their robustness, tolerating geometric errors such as self-intersections and gaps, and effectively handling thin structures~\cite{hu2019triwild,hu2018tetwild,hu2020fasttetwild}. They employ an envelope that  encompasses the input triangle soup, performing tetrahedralizaiton based on the envelope and then enhancing mesh quality through optimization. Though robust, they alter boundary geometry, which can introduce errors in applications requiring high-fidelity to the original high-quality triangle meshes. 

\subsection{Optimization Methods}

\paragraph{Laplacian methods}  Optimizing internal vertices is crucial for enhancing mesh quality. Known for their simplicity and efficiency, Laplacian smoothing methods enhance mesh quality by repositioning each vertex to the average location of its neighboring vertices \cite{field1988laplacian,vollmer1999improvedLaplace,sorkine2006differential}. However, they do not guarantee convergence or mesh correctness, and are often combined with other techniques for better results~\cite{freitag1995efficient,freitag1996comparison, klingner2009tetrahedral}. 

\paragraph{Centroidal Voronoi tessellations (CVT)} Originally proposed by \cite{du1999centroidal}, who positioned generating points of each Voronoi cell at its centroid, resulting in high-quality Voronoi diagrams. The CVT energy, a \(C^2\) continuous but non-convex function \cite{liu2009centroidal}, has been adapted for triangulation optimization \cite{yan2013efficient} and used in various applications including anisotropic mesh generation \cite{telsang2022computation, du2005anisotropic, levy2010p, ekelschot2018robust, ekelschot2019parallel, xiao2018optimal}. Its use in tetrahedral mesh optimization can, however, produce many sliver elements, affecting mesh quality and stability \cite{chen2014revisiting, 2004TOG_ODT}.

\paragraph{Optimal Delaunay triangulations (ODT)} \cite{chen2004ODT} introduced ODT for enhancing the quality of triangular meshes, focusing on minimizing an energy function tied to the weighted interpolation error \cite{chen2004ODT,chen2004ODT2}. \cite{2004TOG_ODT} expanded ODT's application to tetrahedral mesh optimization. Building on this, \cite{2009NODT} developed the natural ODT method (NODT), optimizing both internal and boundary vertices of tetrahedral meshes. \cite{chen2011efficient} introduced several local and global optimization techniques to minimize the ODT energy and also proposed the centroidal patch triangulation (CPT) method. Due to the ODT energy's nonlinearity and non-smoothness, achieving global optimal solutions is challenging \cite{2022degreeODT,chen2011efficient,chen2014revisiting}. Nevertheless, recent research has marked significant progress in overcoming these challenges \cite{chen2014revisiting,weng2024global}.

\section{Problem Formulation}
\label{sec:formulation}

Typically, tetrahedral elements of poor quality are characterized by non-uniform shapes, which lead to significant volume variations across tetrahedra, irregular vertex connectivity, and extreme dihedral angles that approach either $0^\circ$ or $180^\circ$. 
To enhance mesh quality, our approach aims to reduce volume variations, regularize vertex connectivity, and narrow the range of dihedral angles to avoid extremes. 

Towards this goal, we propose a weighted squared volume energy, which is proved to be convex, for optimizing each movable internal vertex. When minimized with constant weights, this function promotes uniform volumes across tetrahedra. Adjusting the weights to account for local geometry results in uniform dihedral angles within the mesh. 
After updating the vertex locations, we adopt flip operations to improve connectivity. We also use split and collapse operations to eliminate local short and long edges.


\subsection{Weighted Squared Volume Energy}\label{subsec:Theorem} 

Let \( \mathcal{S} \) be a closed, genus-0 surface that bounds a region, and let \(\mathbf{V}: \Omega \rightarrow \mathbb{R}^3\) be a \( C^2 \)-continuous parametric function representing this region, where \(\Omega = [0,1]^3\). Denote by $\mathbf{V}_x, \mathbf{V}_y, \mathbf{V}_z$ the partial derivatives of $\bf V$ with respect to the coordinates $x$, $y$ and $z$, respectively. We define the functional space $\mathbb{V}$ as
\[
\mathbb{V} = \left\{ \mathbf{V} \in C^2\left(\Omega, \mathbb{R}^3\right) \left |  ~~\mathbf{V}|_{\partial\Omega} = \mathcal{S}~\mathrm{and}~\mathbf{V}_x \cdot (\mathbf{V}_y \times \mathbf{V}_z) \neq 0  \right.\right\}
\]
to represent all regular and $C^2$-continuous parametric representations of the solid. The condition $\mathbf{V}_x \cdot (\mathbf{V}_y \times \mathbf{V}_z) \neq 0$ indicates that the volume element is non-zero, thereby the parametric representation is regular. We introduce an energy functional \(E(\mathbf{V}): \mathbb{V} \rightarrow \mathbb{R}\) defined over the functional space \(\mathbb{V}\).

\begin{definition}\label{energy-definition}
Let \(\rho: \Omega \to \mathbb{R}_+\) be a $C^2$-continuous density function defined over the domain \(\Omega\). The weighted squared volume energy \(E(\mathbf{V})\) is defined as the volume integral of weighted squared volumes:
\begin{equation}
E(\mathbf{V}) = \iiint_\Omega \rho \left(\mathbf{V}_x \cdot \left(\mathbf{V}_y \times \mathbf{V}_z\right)\right)^2 \mathrm{d}x\mathrm{d}y\mathrm{d}z.
\end{equation}\label{eqn:continuous-energy}
\end{definition}

Applying the variational principle, we show that the critical points of the energy functional are characterized by the function $\bf V$ such that the volume element remains constant. 

\begin{theorem}\label{Theorem 1}
At the critical points of the weighted squared volume energy $E(\mathbf{V})$, the quantity $\rho\left(\mathbf{V}_x\cdot(\mathbf{V}_y\times\mathbf{V}_z)\right)$ is constant across $\Omega$.
\end{theorem}

This theorem implies that minimizing this energy functional promotes uniform weighted volume elements. See the supplementary material for the proof.

\textbf{Remark 1.} Our weighted squared volume energy is an extension of the weighted squared area energy, as proposed by~\cite{twosimple} for remeshing surfaces. In the context of 2-manifolds, Renka proved that minimizers of the squared area energy lead to uniformly parameterized minimal surfaces. Building on this result, he developed an algorithm to enhance the quality of triangle meshes. Inspired by Renka's work, we adapt his approach by replacing squared areas with squared volumes, calculated using scalar triple products.


\subsection{Convex Optimization for Vertex Positioning} 
\label{subsec:convex}

The weighted squared volume energy, as defined in Equation~(\ref{eqn:continuous-energy}), is presented in a continuous form using a volume integral. In this subsection, we discretize it for use on tetrahedral meshes. Let \(\mathcal{M} = (V, E, F, T)\) represent a tetrahedral mesh, where \(V, E, F\) and \(T\) denote the sets of vertices, edges, faces, and tetrahedra, respectively. Consider an internal vertex \(\boldsymbol{v}_i \in V\). Denote by \(S_i\) the space defined by the 1-ring neighborhood of \(\boldsymbol{v}_i\), which consists of $n$ tetrahedra incident to \(\boldsymbol{v}_i\). Consider the $m$-th tetrahedron in \( S_i \), delineated by the cyclically ordered vertices \(\boldsymbol{v}_i\), \(\boldsymbol{v}_{mj}\), \(\boldsymbol{v}_{mk}\), and \(\boldsymbol{v}_{ml}\). The volume of this tetrahedron, denoted as $\overline{V}_m(\boldsymbol{v}_i)$, can be computed using scalar triple product:
\[
\overline{V}_m(\boldsymbol{v}_i) = \frac{1}{6} (\boldsymbol{v}_{mk} - \boldsymbol{v}_{ml}) \cdot \left( (\boldsymbol{v}_i - \boldsymbol{v}_{ml}) \times (\boldsymbol{v}_{mj} - \boldsymbol{v}_{ml}) \right) .
\]
Assuming the density function 
$\rho$ varies insignificantly within each tetrahedron, we approximate the density for the entire tetrahedron using the value defined at its centroid. We then discretize the energy functional \(E(\mathbf{V})\) for the space \(S_i\) by summing the squared volumes of the tetrahedra within \(S_i\): 
\begin{equation}\label{equ:discrete_energy_Volume_based}
E(\boldsymbol{v}_i) = \sum_{m=1}^{n} \rho_m \overline{V}_m^2(\boldsymbol{v}_i),
\end{equation} where \(\rho_m\) is the density value at the centroid of the \(m\)-th tetrahedron. To ensure the validity of tetrahedral configurations and prevent the presence of negative volumes, we impose the following constraint:
\[
f_m(\boldsymbol{v}_i)=(\boldsymbol{v}_{mk} - \boldsymbol{v}_{ml}) \cdot \left( (\boldsymbol{v}_i - \boldsymbol{v}_{ml}) \times (\boldsymbol{v}_{mj} - \boldsymbol{v}_{ml}) \right) > 0.
\]
We then formulate the optimization problem as:
\begin{equation}
\begin{aligned}
& \underset{\boldsymbol{v}_i}{\min} & & E(\boldsymbol{v}_i) \\
& \text{subject~to~} & & f_m(\boldsymbol{v}_i) > 0, \quad \text{for all } m = 1 \ldots n.
\end{aligned}\label{equ:convexOpt}
\end{equation}

\begin{theorem}\label{thm:convex_optimization}
The optimization problem described in Equation (\ref{equ:convexOpt}) is a convex optimization problem. See the supplementary material for the proof.
\end{theorem}

We employ a Newton solver to address this optimization problem for each internal vertex. By carefully selecting the density function $\rho$, we demonstrate that the new vertex positions induce a tetrahedral mesh with enhanced mesh quality. Specifically, we adopt two different weighting strategies in the weighted squared volume energy: a constant weight and a weight adaptive to local geometry. These strategies are detailed as follows. 

In the continuous setting, setting the density value \(\rho_m \equiv 1\) results in all volume elements in the integral having equal volumes. However, in the discrete setting, the continuous volume integral within the domain \(\Omega\) is approximated by the summation of the volumes of a finite number of individual tetrahedra. Due to discretization errors, achieving exact volume uniformity is nearly impossible. Typically, the volumes of the tetrahedra follow a normal distribution. Experimental results confirm that as the resolution of the tetrahedral mesh increases, the standard deviation of this normal distribution decreases rapidly, as illustrated in Figure~\ref{fig:volume_distribution}, justifying the practical use of constant density for achieving uniform volume. 

\begin{figure}[htbp]
\centering
\setlength{\abovecaptionskip}{0pt}
\setlength{\belowcaptionskip}{0pt}
\subcaptionbox*{}{%
\begin{minipage}[t]{0.2\linewidth}
\centering
\includegraphics[width=\linewidth]{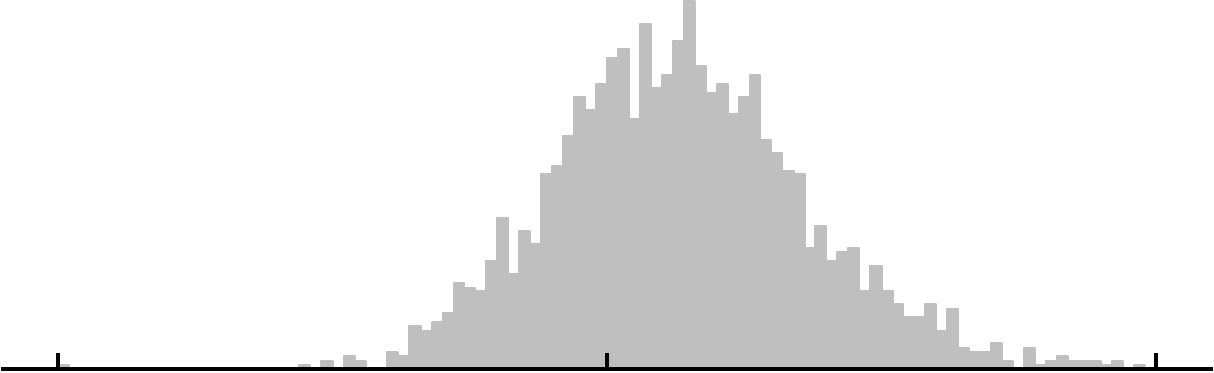}
\end{minipage}
}%
\subcaptionbox*{}{%
\begin{minipage}[t]{0.2\linewidth}
\centering
\includegraphics[width=\linewidth]{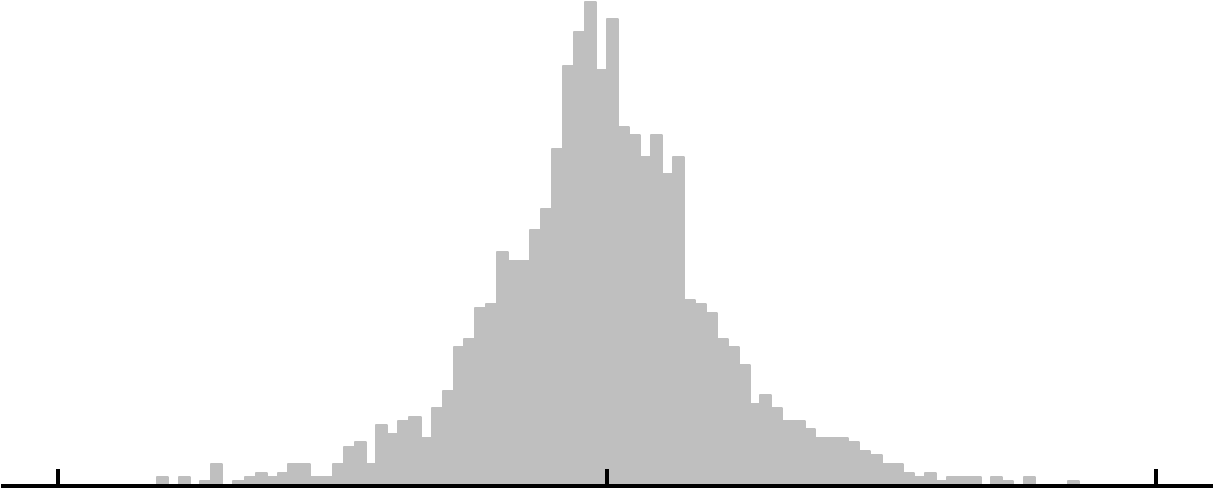}
\end{minipage}
}%
\subcaptionbox*{}{%
\begin{minipage}[t]{0.2\linewidth}
\centering
\includegraphics[width=\linewidth]{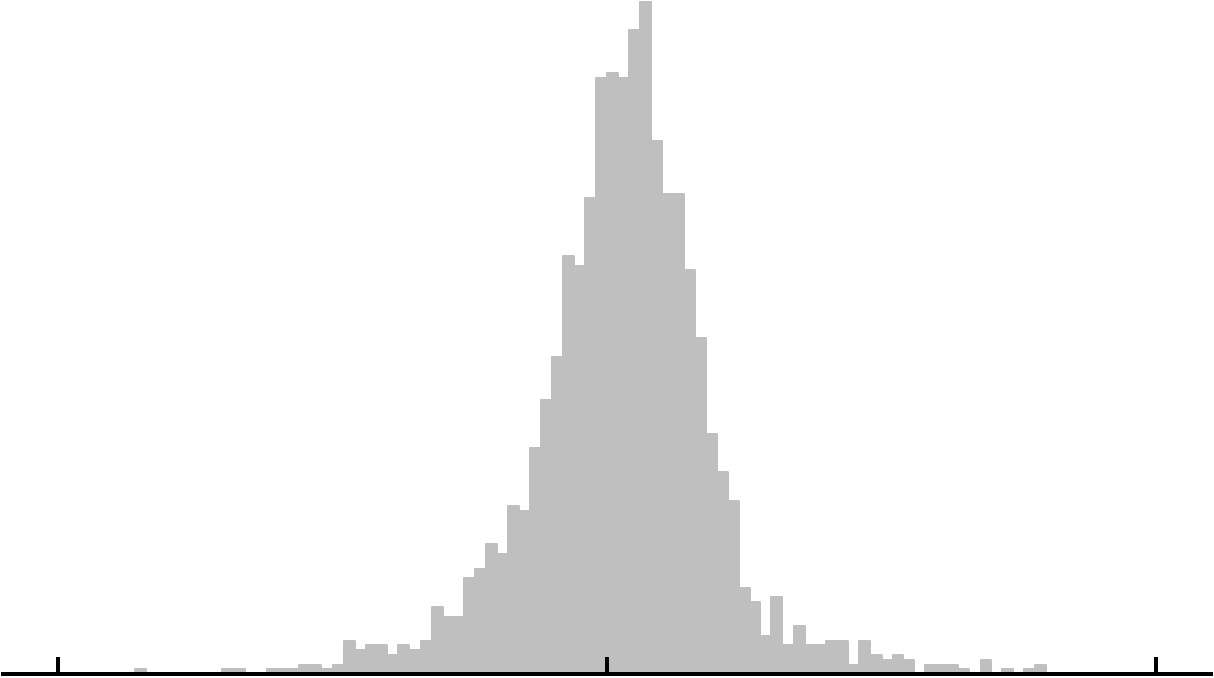}
\end{minipage}
}%
\subcaptionbox*{}{%
\begin{minipage}[t]{0.2\linewidth}
\centering
\includegraphics[width=\linewidth]{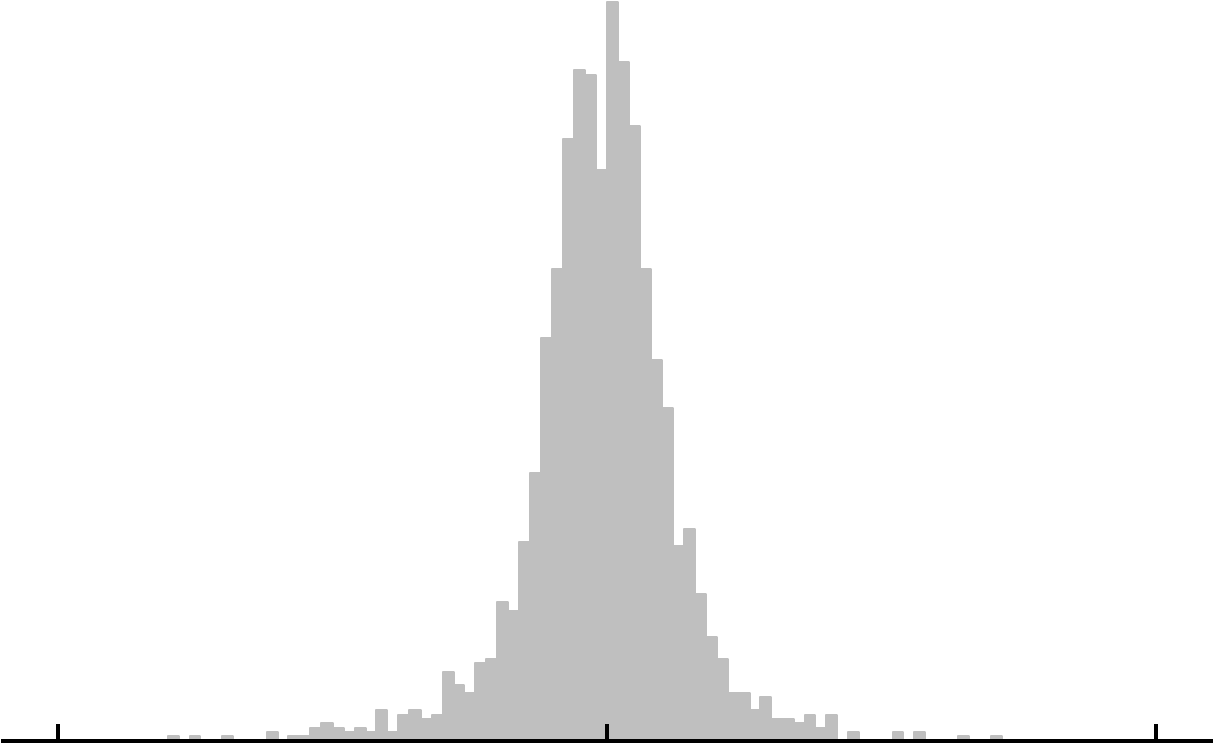}
\end{minipage}
}%
\vspace{-5mm}
\subcaptionbox*{}{%
\begin{minipage}[t]{0.2\linewidth}
\centering
\begin{overpic}[width=\linewidth]{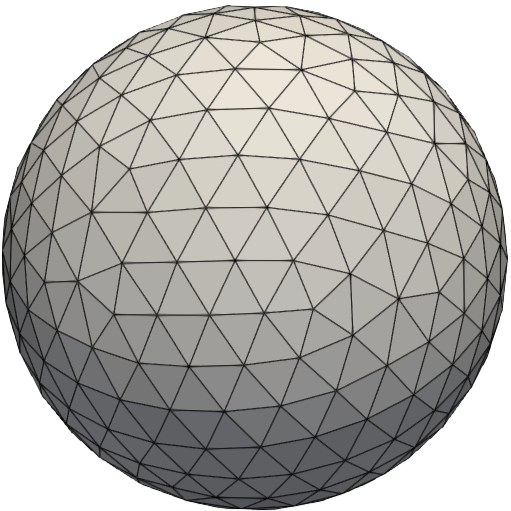}
\put(70,92){\color{black}\small 2K}
\end{overpic}
\end{minipage}
}%
\subcaptionbox*{}{%
\begin{minipage}[t]{0.2\linewidth}
\centering
\begin{overpic}[width=\linewidth]{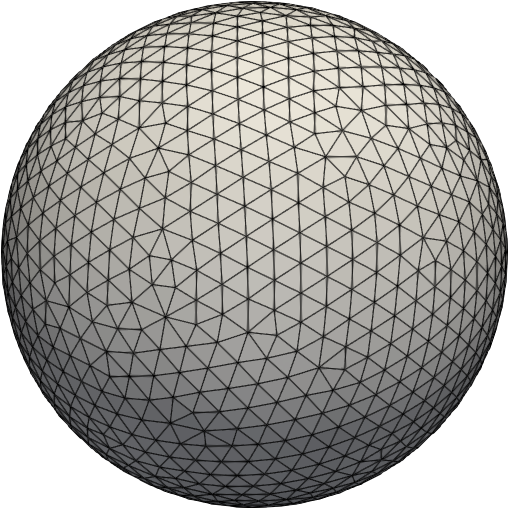}
\put(70,92){\color{black}\small 17K}
\end{overpic}
\end{minipage}
}%
\subcaptionbox*{}{%
\begin{minipage}[t]{0.2\linewidth}
\centering
\begin{overpic}[width=\linewidth]{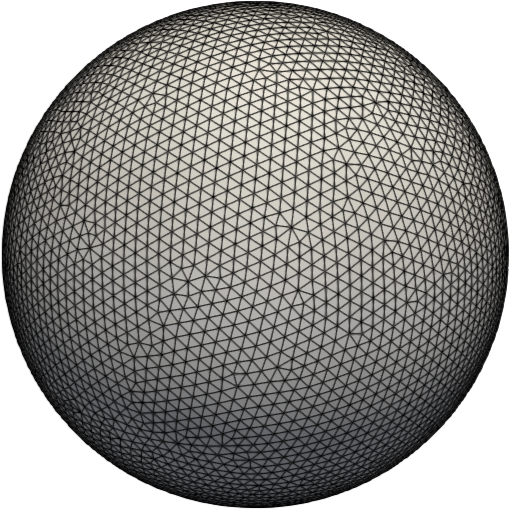}
\put(70,92){\color{black}\small 126K}
\end{overpic}
\end{minipage}
}%
\subcaptionbox*{}{%
\begin{minipage}[t]{0.2\linewidth}
\centering
\begin{overpic}[width=\linewidth]{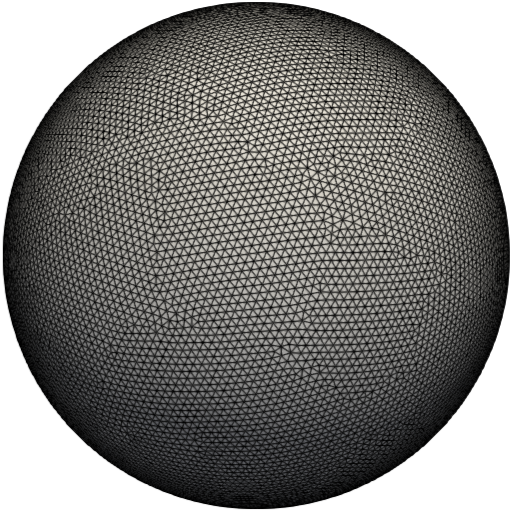}
\put(70,92){\color{black}\small 559K}
\end{overpic}
\end{minipage}
}%
\vspace{-6mm}
\subcaptionbox*{}{%
\begin{minipage}[t]{0.2\linewidth}
\centering
\includegraphics[width=\linewidth]{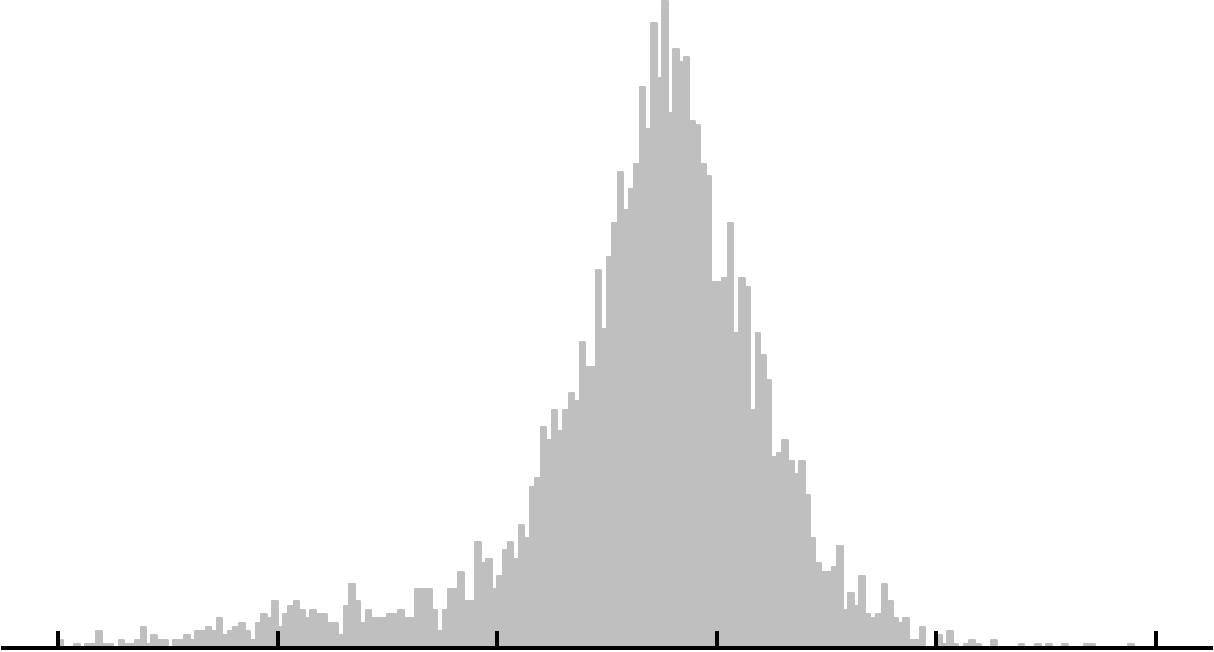}
\end{minipage}
}%
\subcaptionbox*{}{%
\begin{minipage}[t]{0.2\linewidth}
\centering
\includegraphics[width=\linewidth]{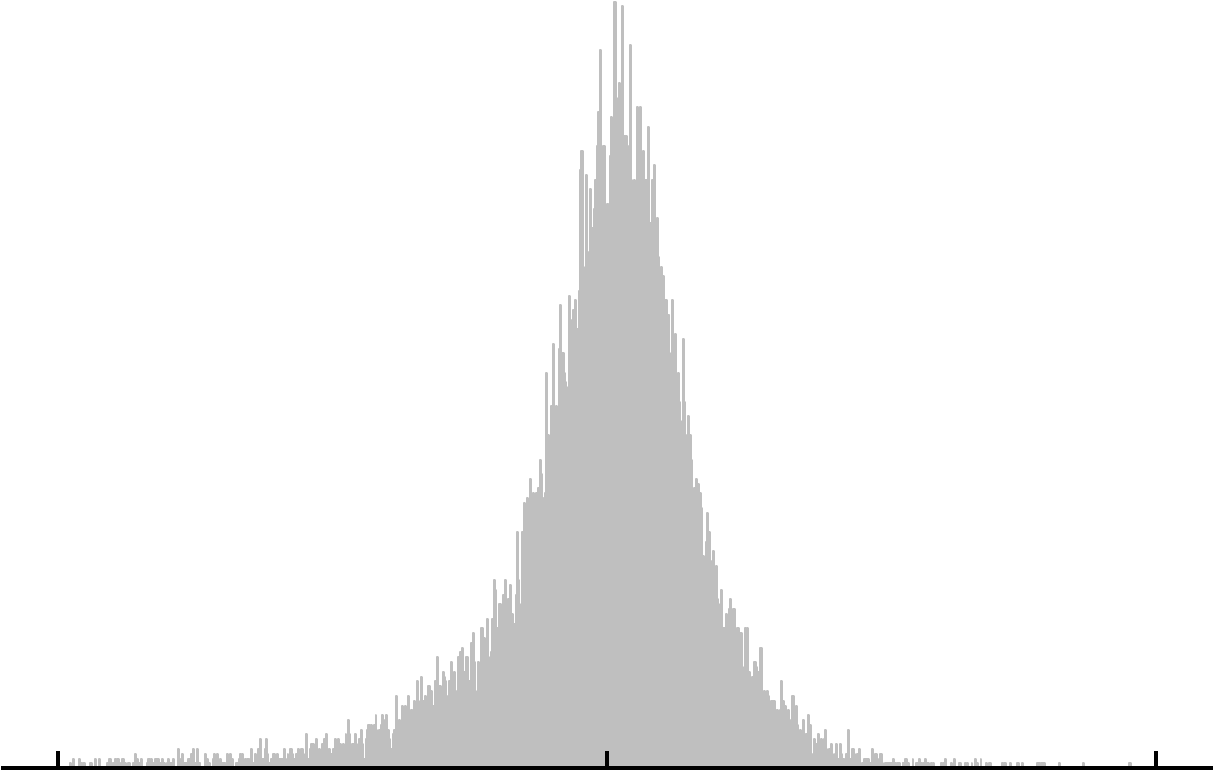}
\end{minipage}
}%
\subcaptionbox*{}{%
\begin{minipage}[t]{0.2\linewidth}
\centering
\includegraphics[width=\linewidth]{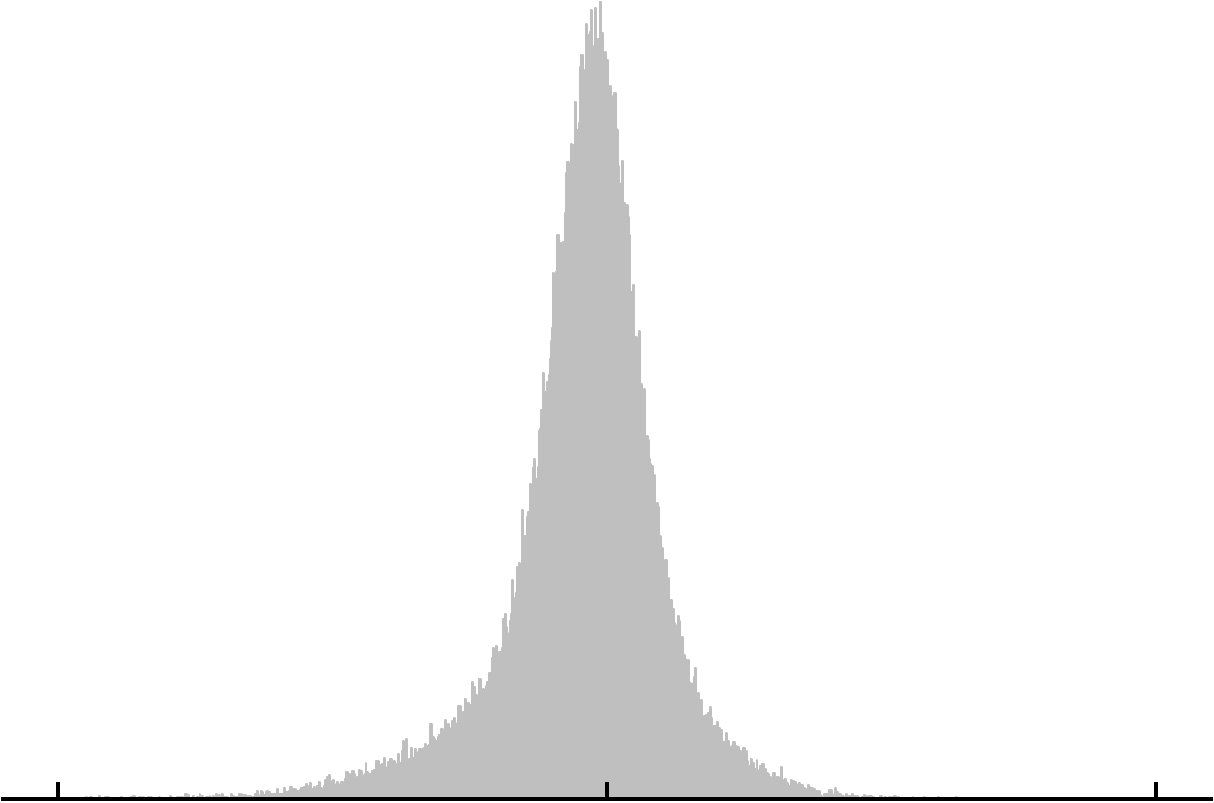}
\end{minipage}
}%
\subcaptionbox*{}{%
\begin{minipage}[t]{0.2\linewidth}
\centering
\includegraphics[width=\linewidth]{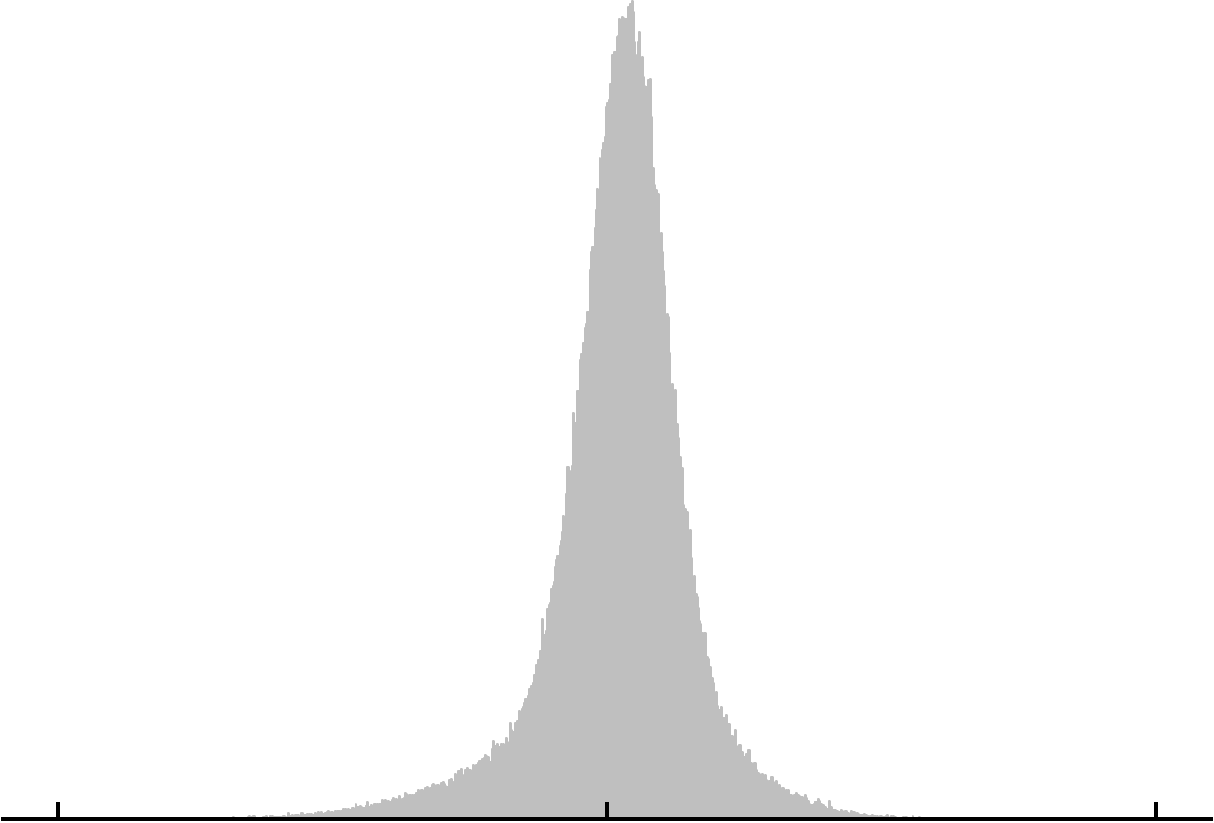}
\end{minipage}
}%
\vspace{-5mm}
\subcaptionbox*{}{%
\begin{minipage}[t]{0.2\linewidth}
\centering
\begin{overpic}[width=\linewidth]{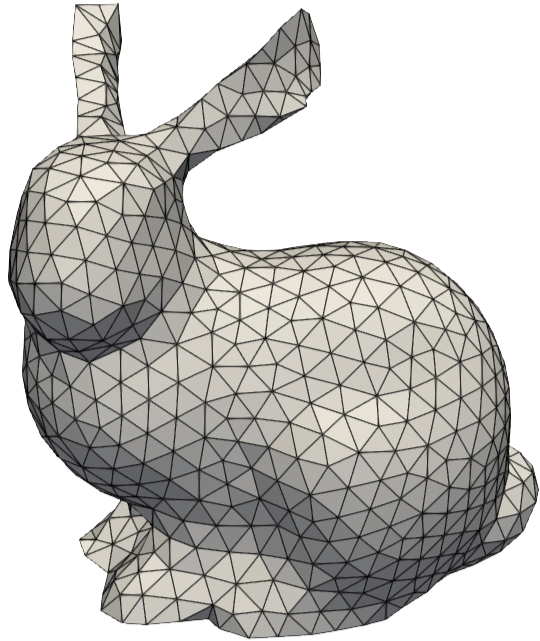}
\put(60,90){\color{black}\small 7K}
\end{overpic}
\end{minipage}
}%
\subcaptionbox*{}{%
\begin{minipage}[t]{0.2\linewidth}
\centering
\begin{overpic}[width=\linewidth]{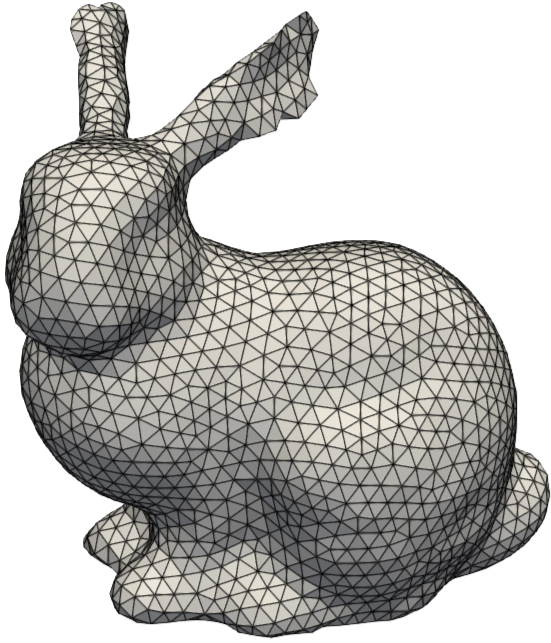}
\put(60,90){\color{black}\small 25K}
\end{overpic}
\end{minipage}
}%
\subcaptionbox*{}{%
\begin{minipage}[t]{0.2\linewidth}
\centering
\begin{overpic}[width=\linewidth]{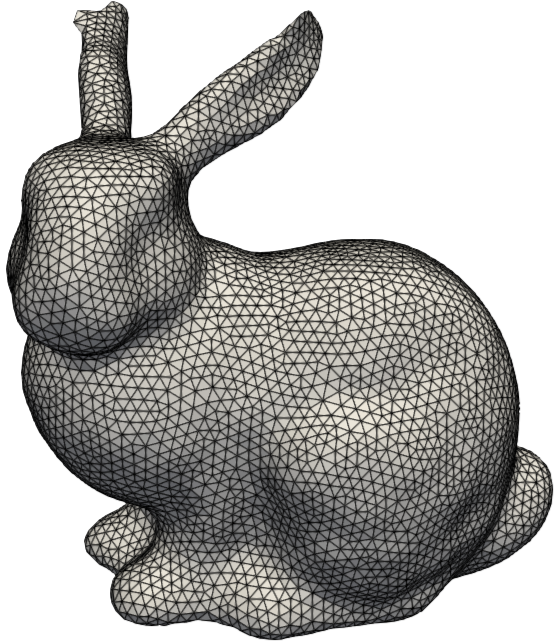}
\put(60,90){\color{black}\small 98K}
\end{overpic}
\end{minipage}
}%
\subcaptionbox*{}{%
\begin{minipage}[t]{0.2\linewidth}
\centering
\begin{overpic}[width=\linewidth]{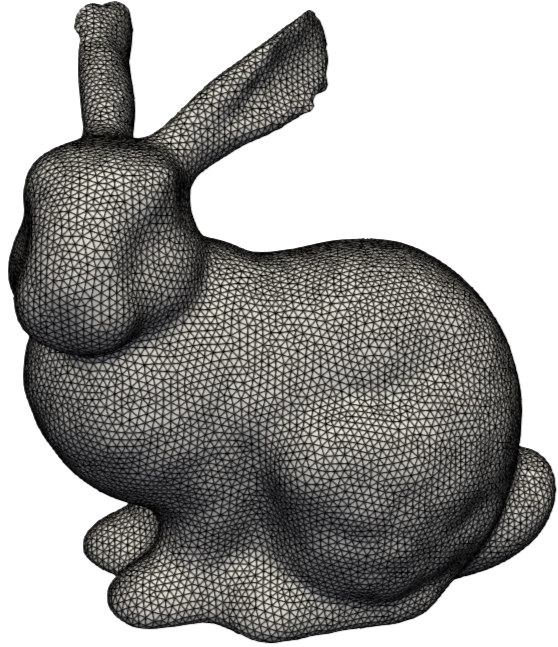}
\put(60,90){\color{black}\small 316K}
\end{overpic}
\end{minipage}
}%
\vspace{-0.1in}
\caption{Minimizing volume-oriented energy promotes uniformity in the volumes of tetrahedral elements. As the resolution of the boundary mesh increases, the generated tetrahedral meshes also exhibit a greater number of degrees of freedom. Consequently, this leads to improved uniformity in tetrahedral volumes, as demonstrated by a reduced standard deviation in the volume distribution.}\label{fig:volume_distribution}
\end{figure}

However, pursuing uniform tetrahedral volumes alone does not guarantee a high-quality tetrahedral mesh, as this volume-centric approach does not ensure uniformity in tetrahedral shapes. To reduce shape variation, we introduce a geometry-adaptive weight \(\rho_m = \frac{1}{A_m(\boldsymbol{v}_i)}\), where \(A_m(\boldsymbol{v}_i)\) is the area of the triangle opposing \(\boldsymbol{v}_i\) in the $m$-th tetrahedron. Note that $\frac{\overline{V}_m(\boldsymbol{v}_i)}{A_m(\boldsymbol{v}_i)}$ is proportional to the height with respect to the opposing triangle. This adaptive weight encourages \(\boldsymbol{v}_i\)
to align with the perpendicular foot of the opposite triangle, effectively equalizing the heights across the tetrahedra in $S_i$. Consequently, the dihedral angle for the edges opposing \(\boldsymbol{v}_i\) tend to be equalized. 



In our algorithm, we refer to  the energy using constant weight as \textit{volume-oriented} energy and the energy using adaptive weight $A_m$ as \textit{angle-oriented} energy.

\begin{algorithm}
\setlength{\abovecaptionskip}{0pt}
\setlength{\belowcaptionskip}{0pt}
\caption{WSVM}\label{algo:GlobalOptimization}
\begin{algorithmic}[1] 
\State \textbf{Input:} Closed triangle meshes $\mathcal{M}$ indicating the exterior and interior boundaries, target edge length $t$, and the convergence criteria $\varepsilon$
\State \textbf{Output:} A tetrahedral mesh with $\mathcal{M}$ as boundaries
\State Generate the initial tetrahedral mesh $\mathcal{T}^{(0)}$
\State $\mathcal{T}^{(1)}\gets $\Call{SquaredVolumeMinimizing}{$\mathcal{T}^{(0)}, 1$}
\State $\mathcal{T}^{(2)}\gets$ \Call{SquaredVolumeMinimizing}{$\mathcal{T}^{(1)}, 1/A$}
\State Output $ \mathcal{T}^{(2)}$
\end{algorithmic}
\end{algorithm}

\begin{figure*}[htbp]
\centering
\setlength{\abovecaptionskip}{0pt}
\setlength{\belowcaptionskip}{0pt}
\subcaptionbox{Input $\mathcal{M}$\label{fig:pipeline_input}}{%
\begin{minipage}[t]{0.16\linewidth}
\centering
\includegraphics[width=\linewidth]{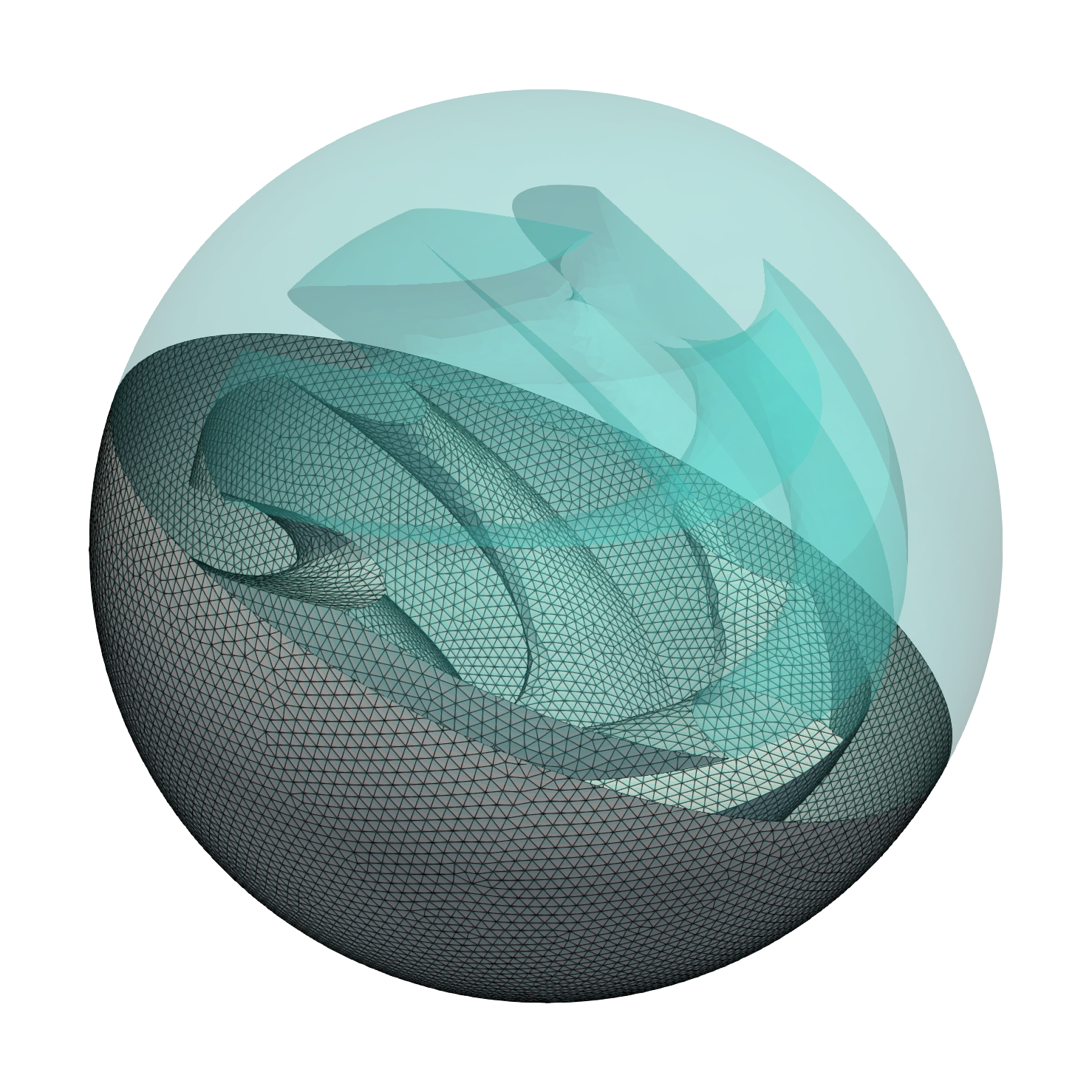}
\end{minipage}}%
\subcaptionbox{Volume-oriented optimizing\label{fig:pipeline_V_based}}
{%
\begin{minipage}[t]{0.34\linewidth} 
\centering
\vspace{0pt} 
\begin{minipage}[t]{0.24\linewidth}
\centering
\includegraphics[width=\linewidth]{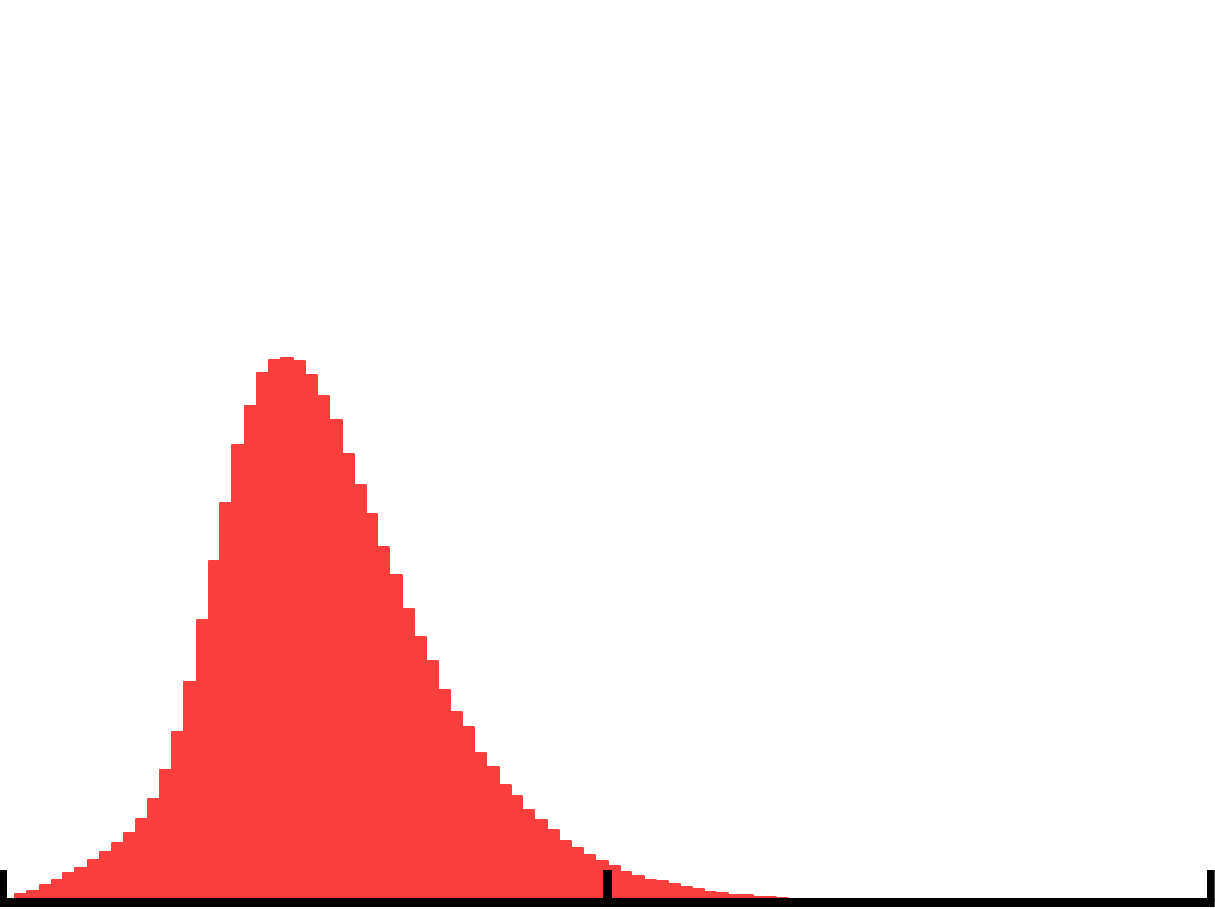}
\end{minipage}%
\begin{minipage}[t]{0.24\linewidth}
\centering
\includegraphics[width=\linewidth]{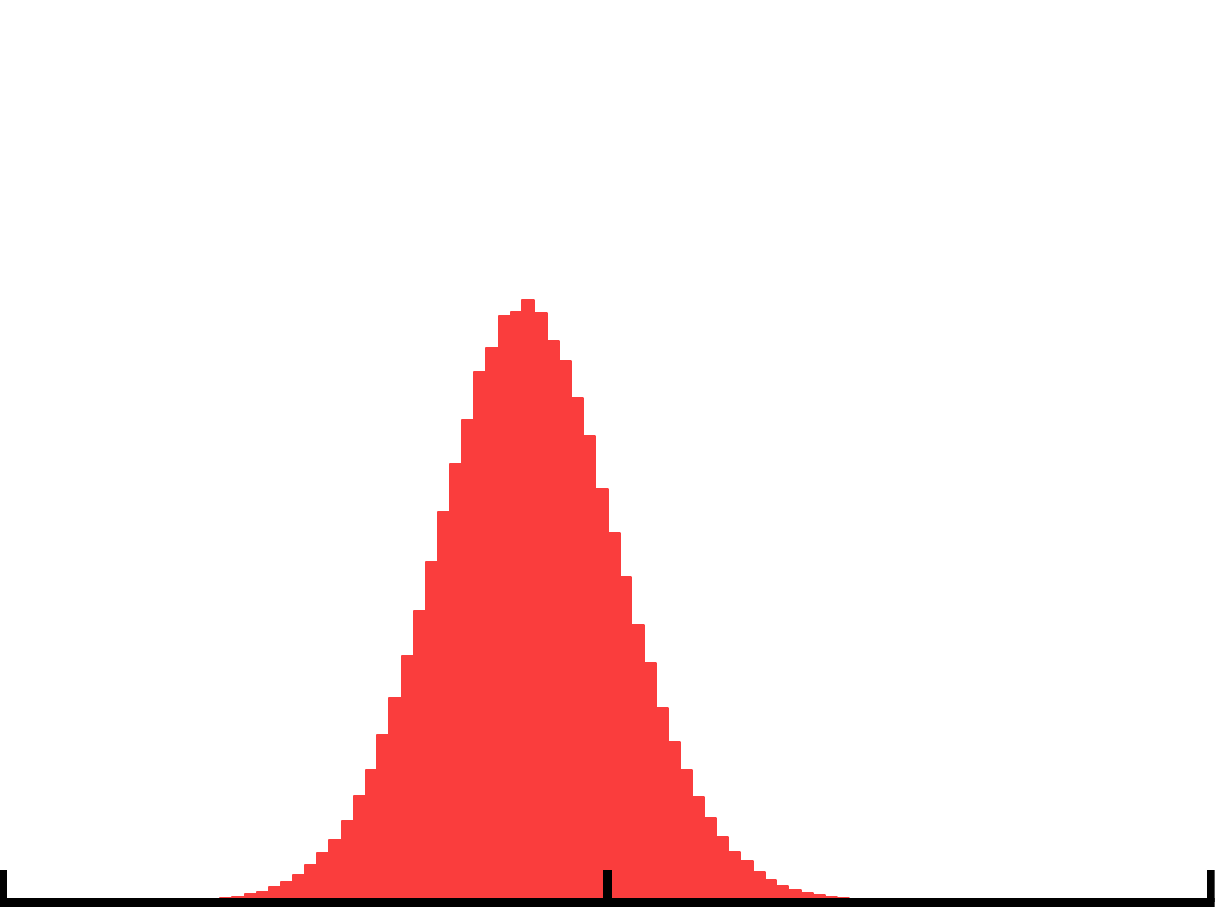}
\end{minipage}
\begin{minipage}[t]{0.24\linewidth}
\centering
\includegraphics[width=\linewidth]{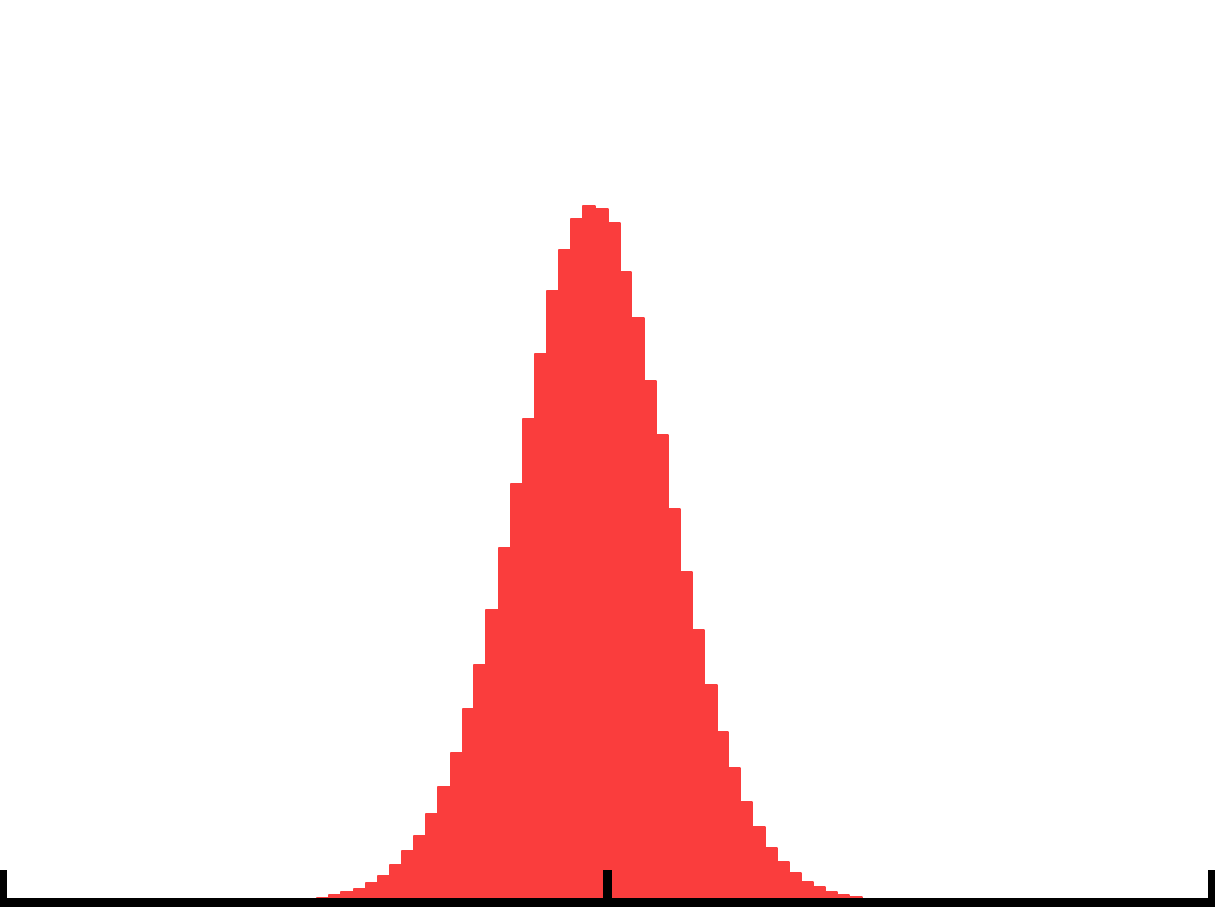}
\end{minipage}
\begin{minipage}[t]{0.24\linewidth}
\centering
\includegraphics[width=\linewidth]{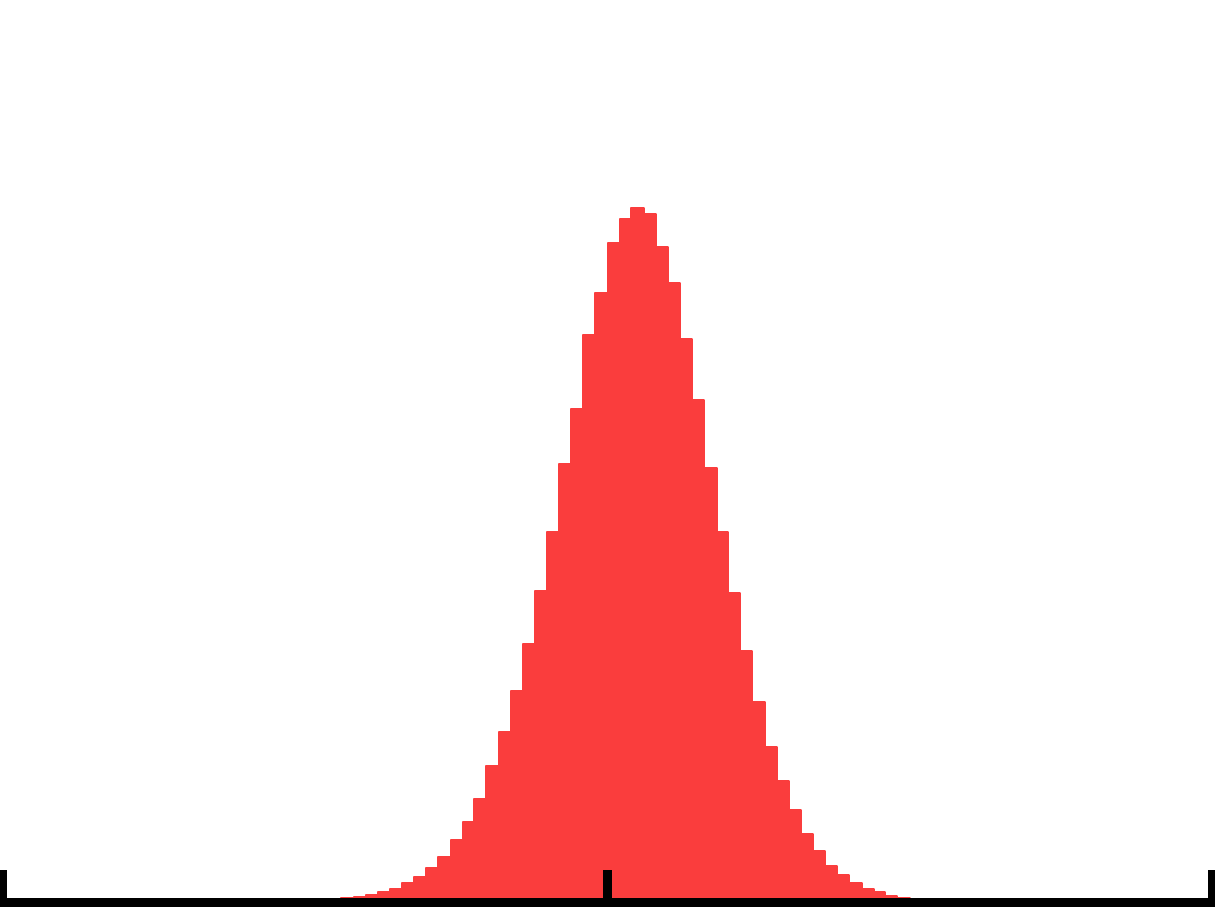}
\end{minipage}
\begin{minipage}[t]{0.24\linewidth}
\centering
\includegraphics[width=\linewidth]{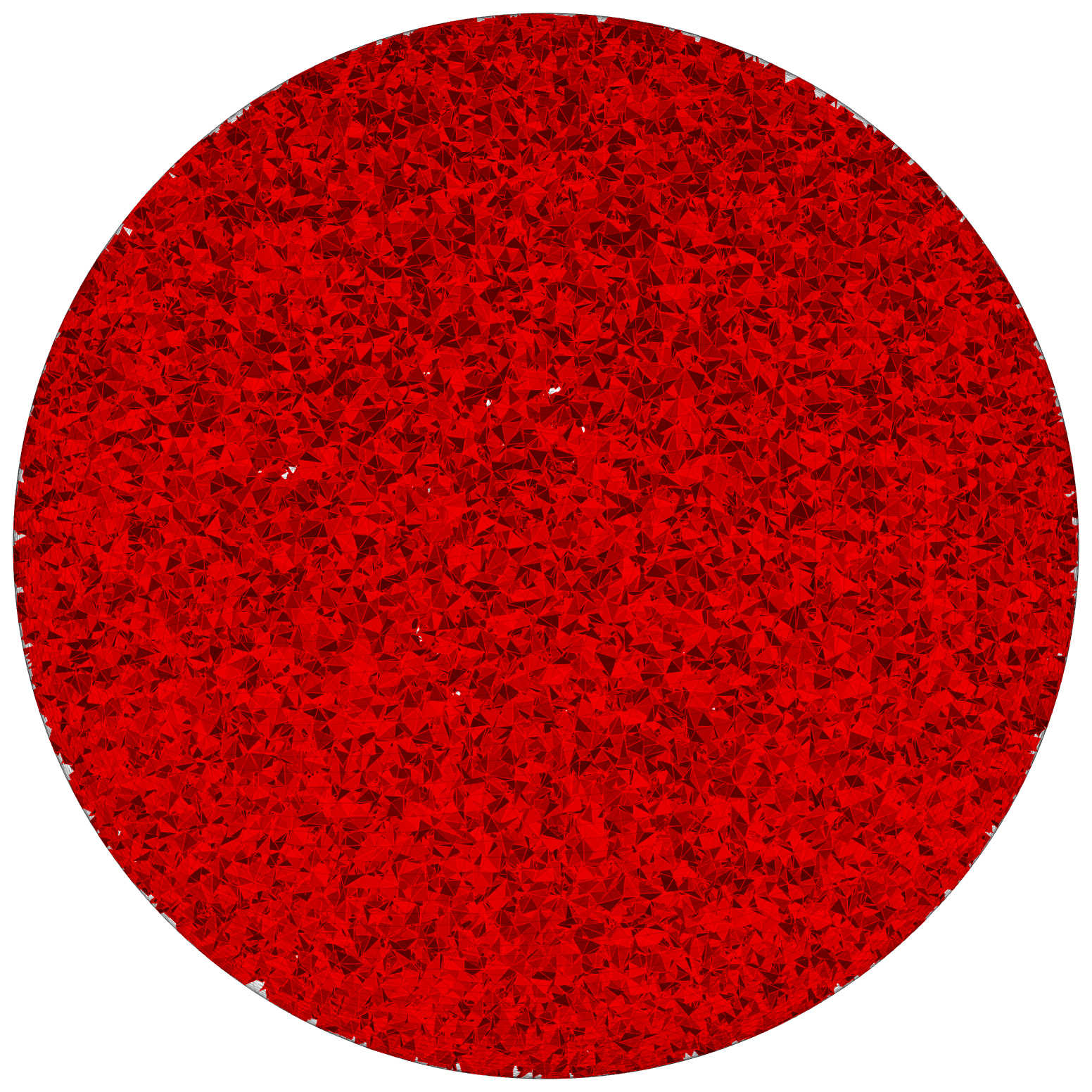}
\end{minipage}%
\begin{minipage}[t]{0.24\linewidth}
\centering
\includegraphics[width=\linewidth]{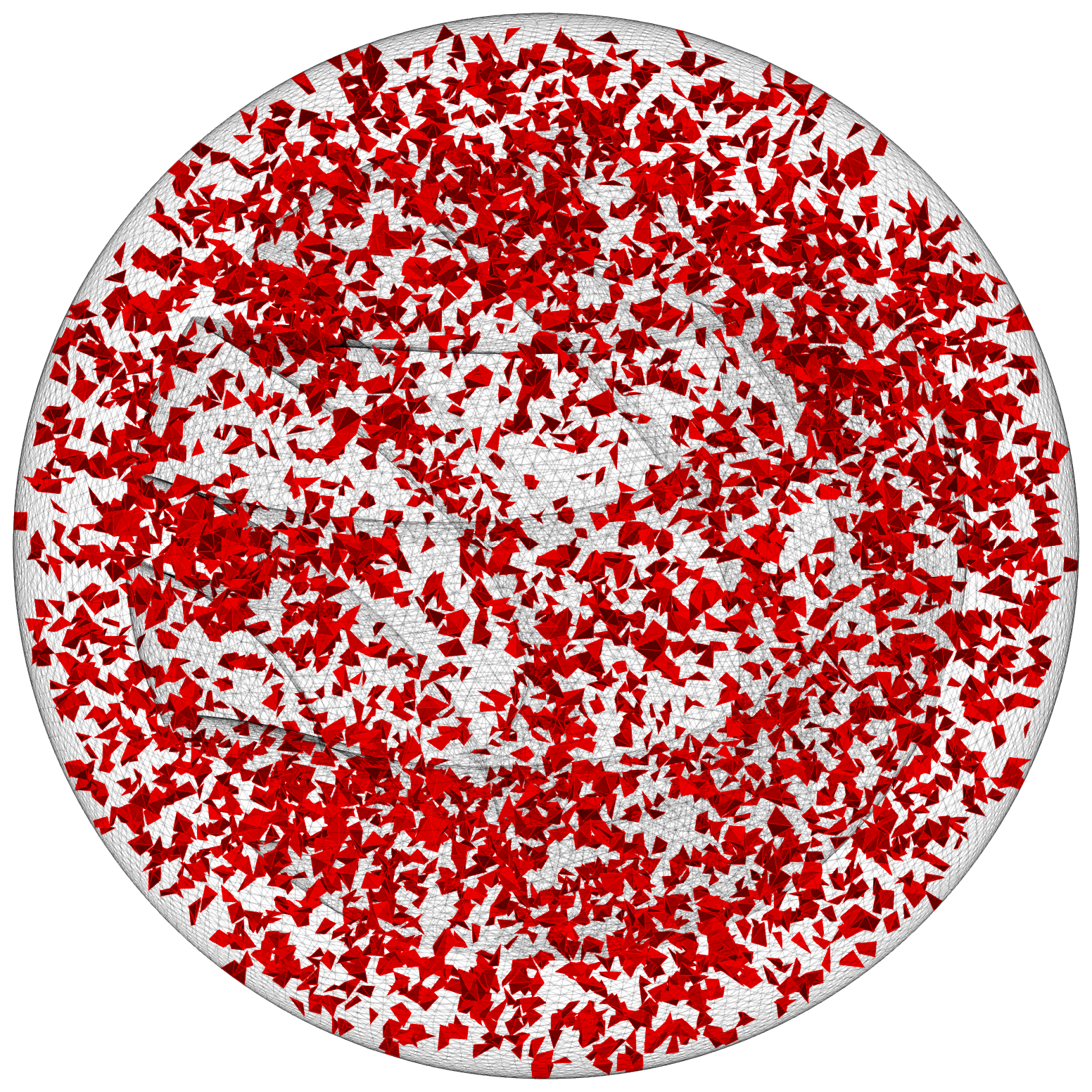}
\end{minipage}
\begin{minipage}[t]{0.24\linewidth}
\centering
\includegraphics[width=\linewidth]{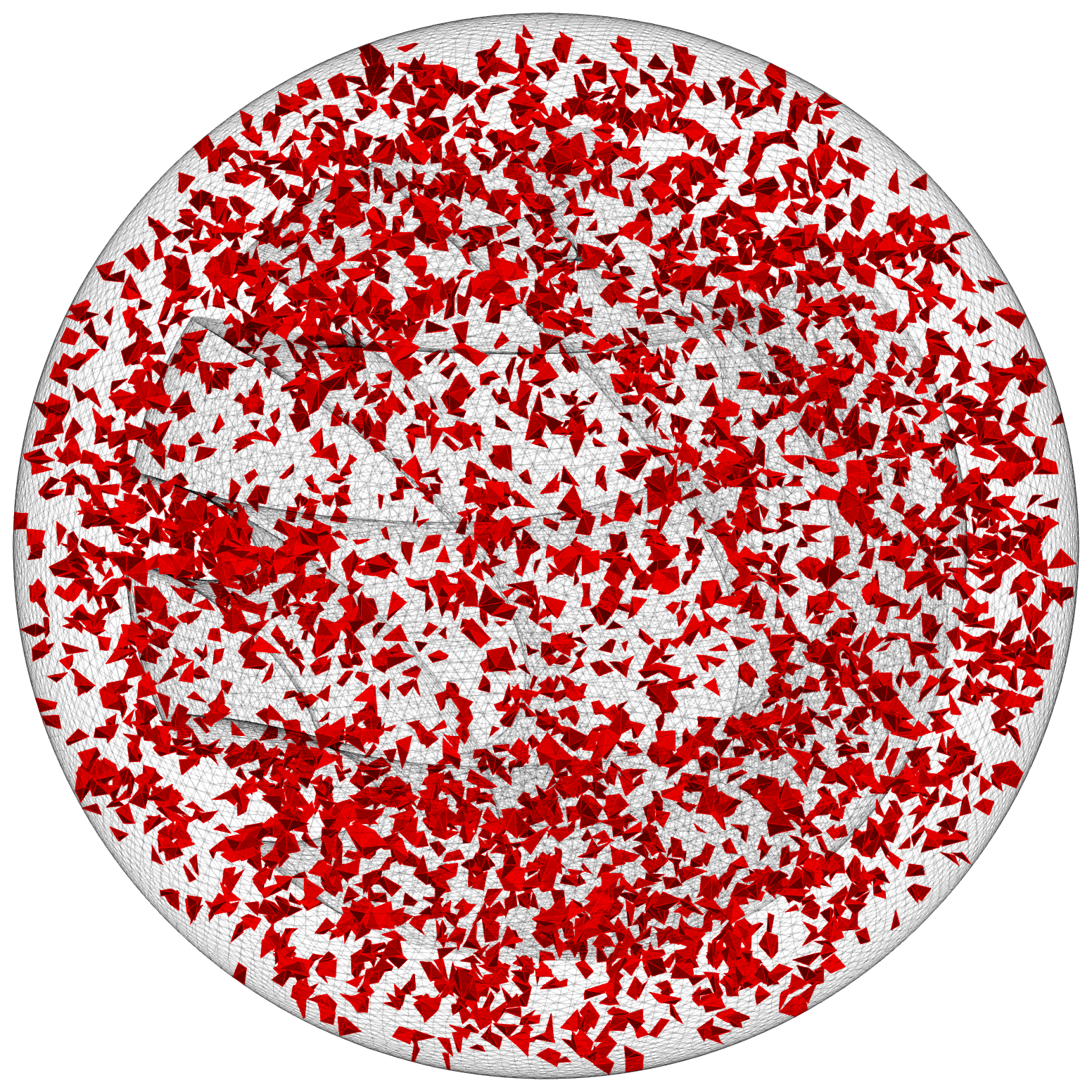}
\end{minipage}
\begin{minipage}[t]{0.24\linewidth}
\centering
\includegraphics[width=\linewidth]{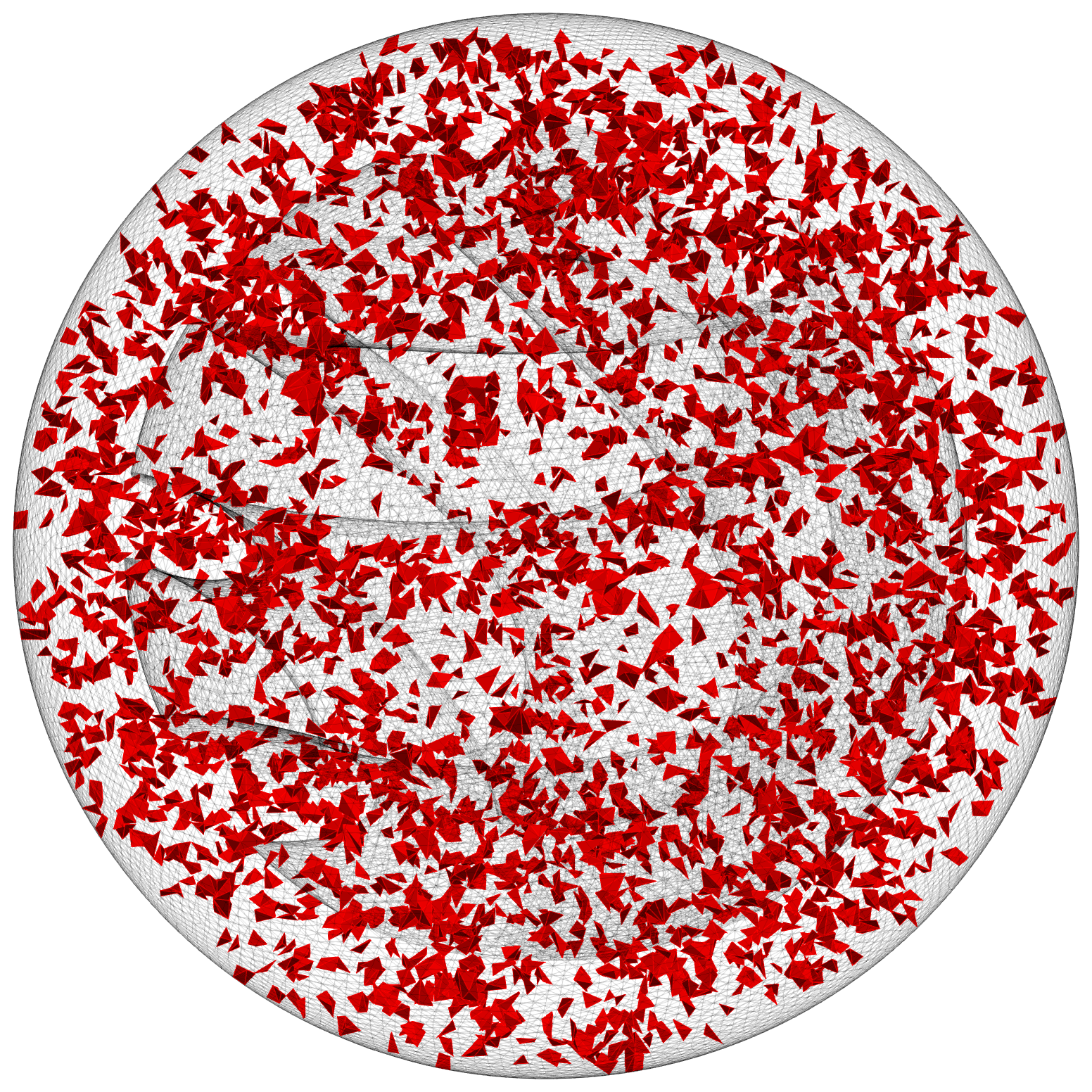}
\end{minipage}
\vspace{0pt} 
\begin{minipage}[t]{0.24\linewidth}
\centering
\includegraphics[width=\linewidth]{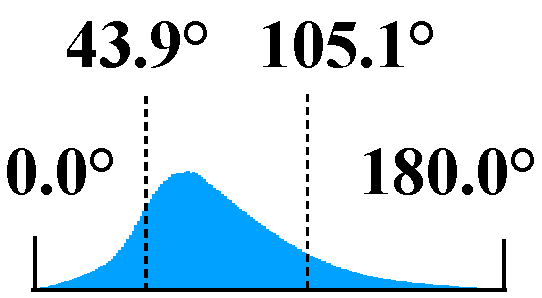}
\end{minipage}%
\begin{minipage}[t]{0.24\linewidth}
\centering
\includegraphics[width=\linewidth]{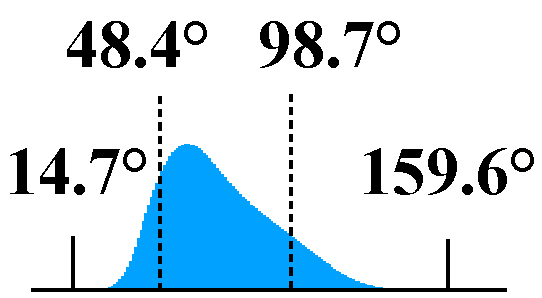}
\end{minipage}
\begin{minipage}[t]{0.24\linewidth}
\centering
\includegraphics[width=\linewidth]{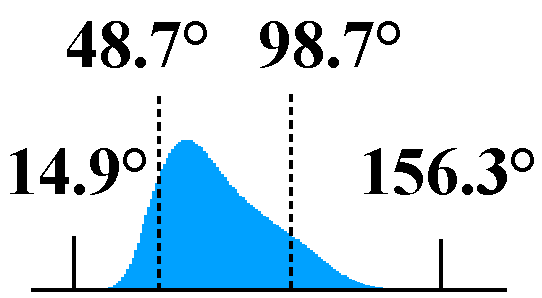}
\end{minipage}
\begin{minipage}[t]{0.24\linewidth}
\centering
\includegraphics[width=\linewidth]{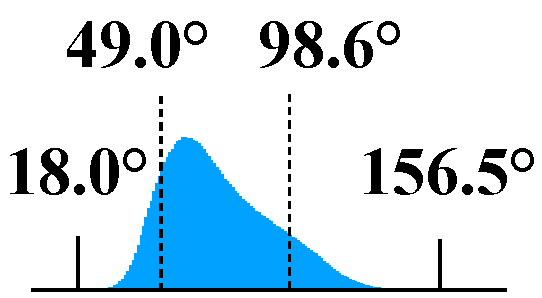}
\end{minipage}
\end{minipage}
}%
\subcaptionbox{Dihedral angle-oriented optimizing\label{fig:pipeline_V+A_based}}{%
\begin{minipage}[t]{0.34\linewidth} 
\centering
\begin{minipage}[t]{0.24\linewidth}
\centering
\includegraphics[width=\linewidth]{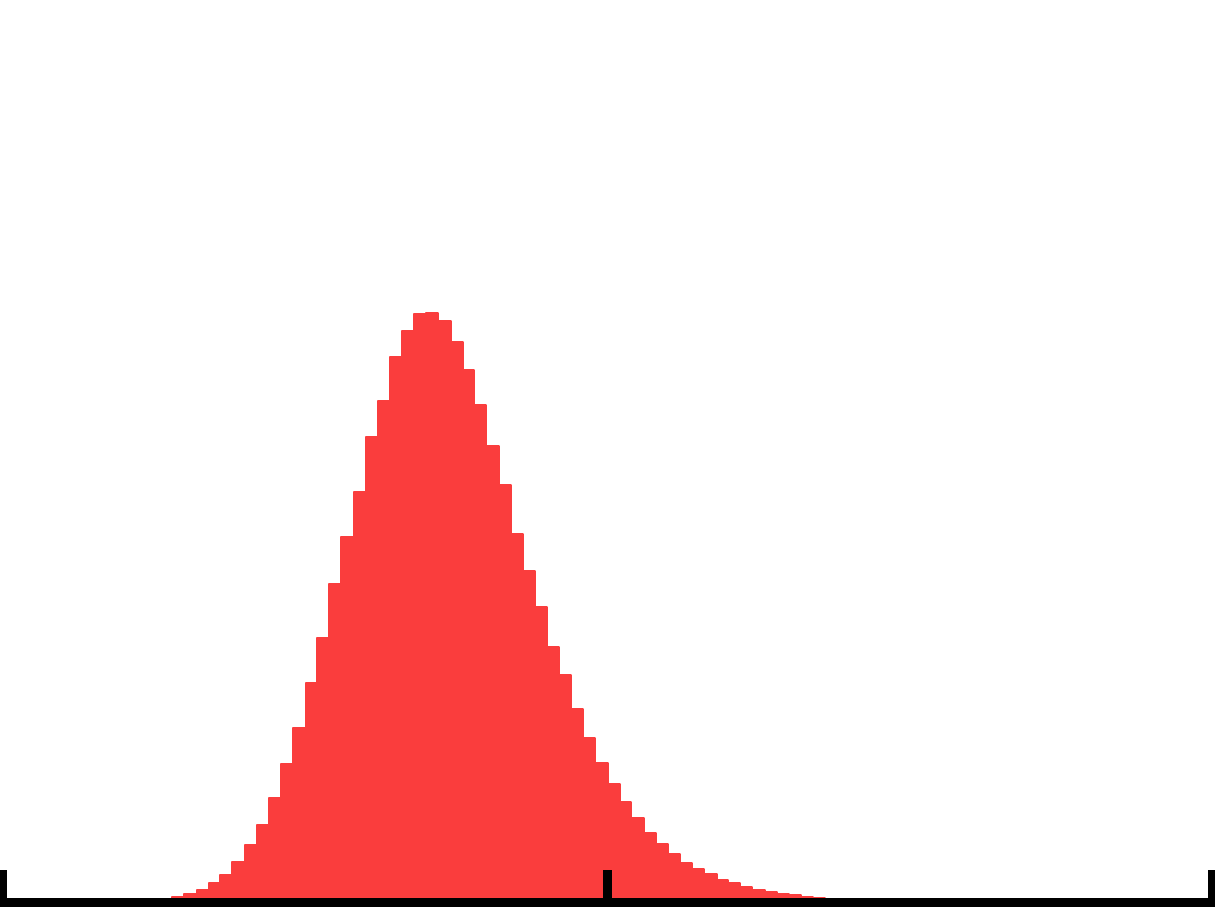}
\end{minipage}%
\begin{minipage}[t]{0.24\linewidth}
\centering
\includegraphics[width=\linewidth]{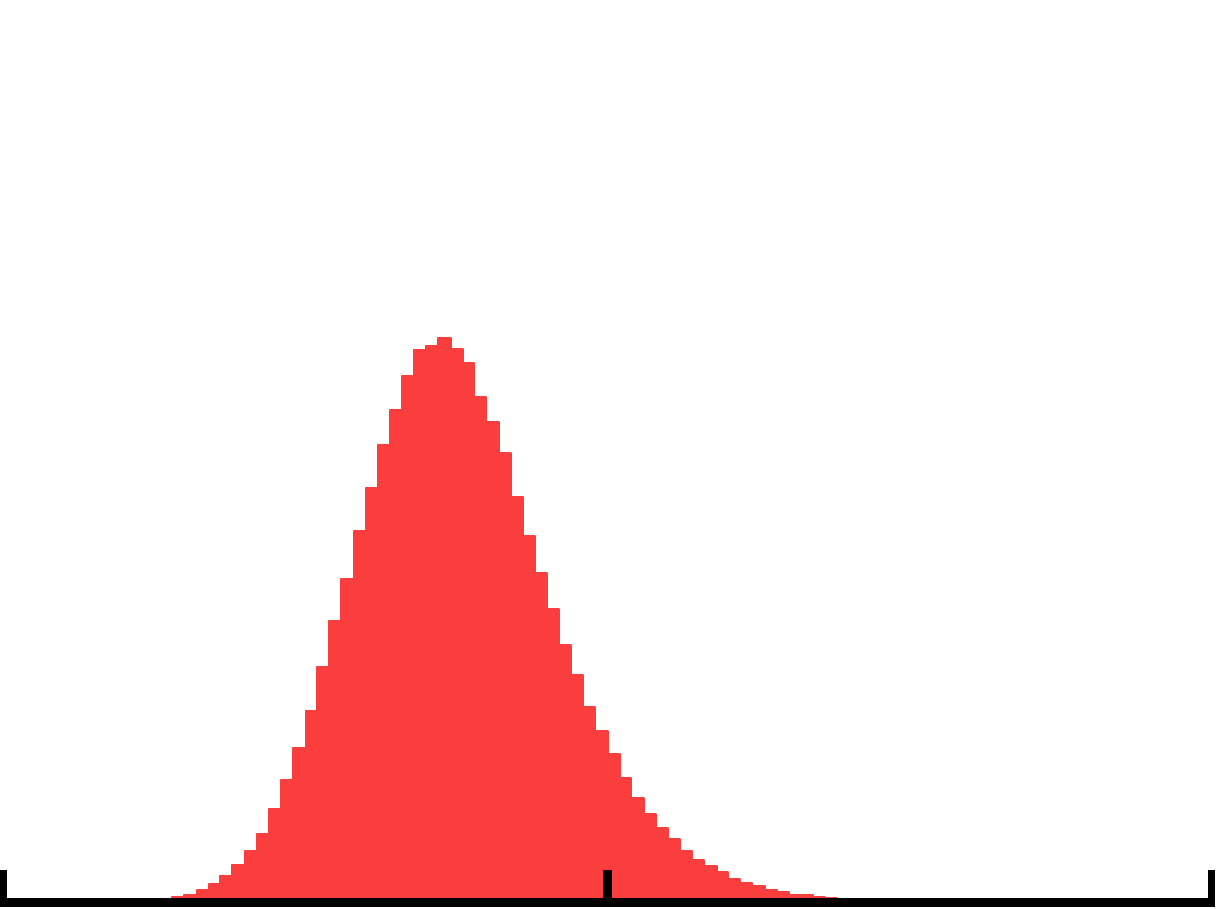}
\end{minipage}
\begin{minipage}[t]{0.24\linewidth}
\centering
\includegraphics[width=\linewidth]{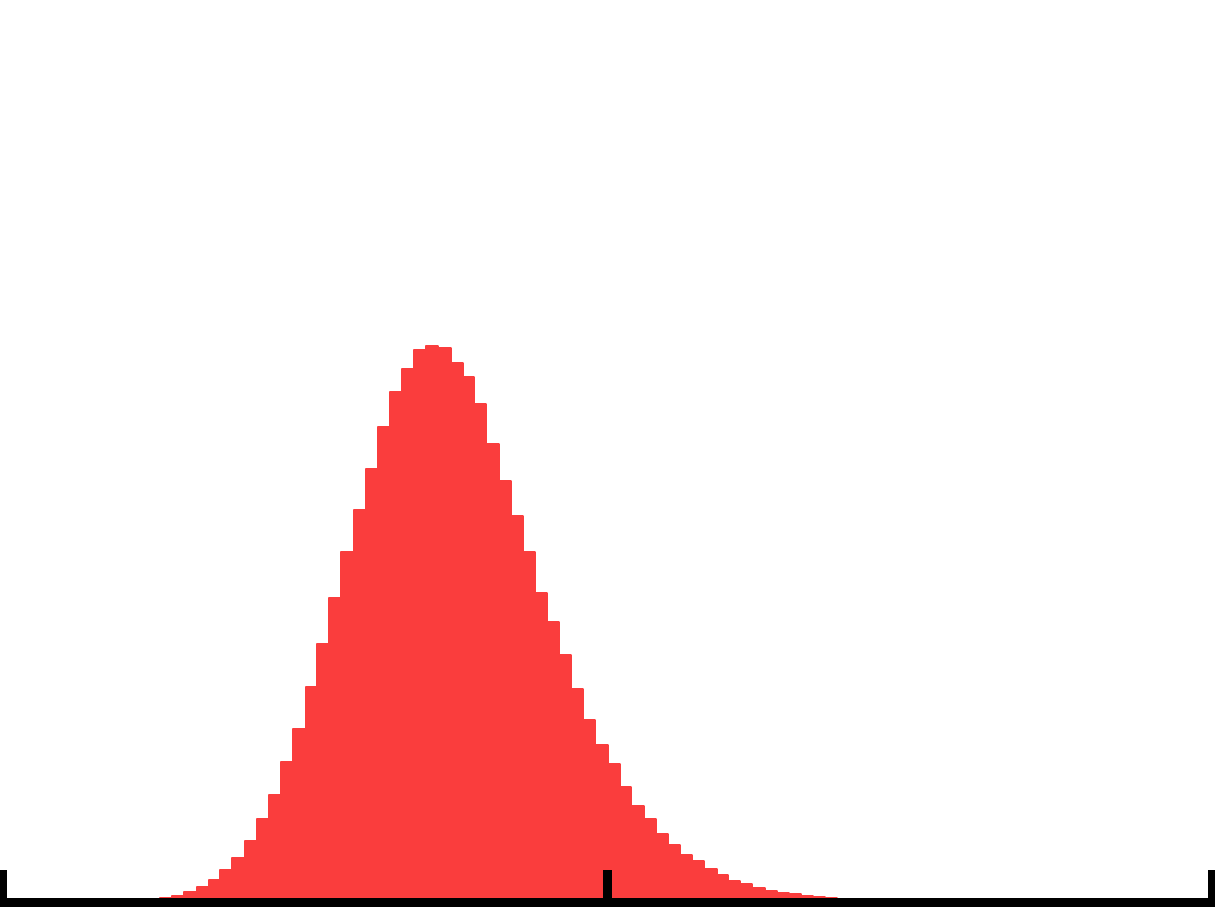}
\end{minipage}
\begin{minipage}[t]{0.24\linewidth}
\centering
\includegraphics[width=\linewidth]{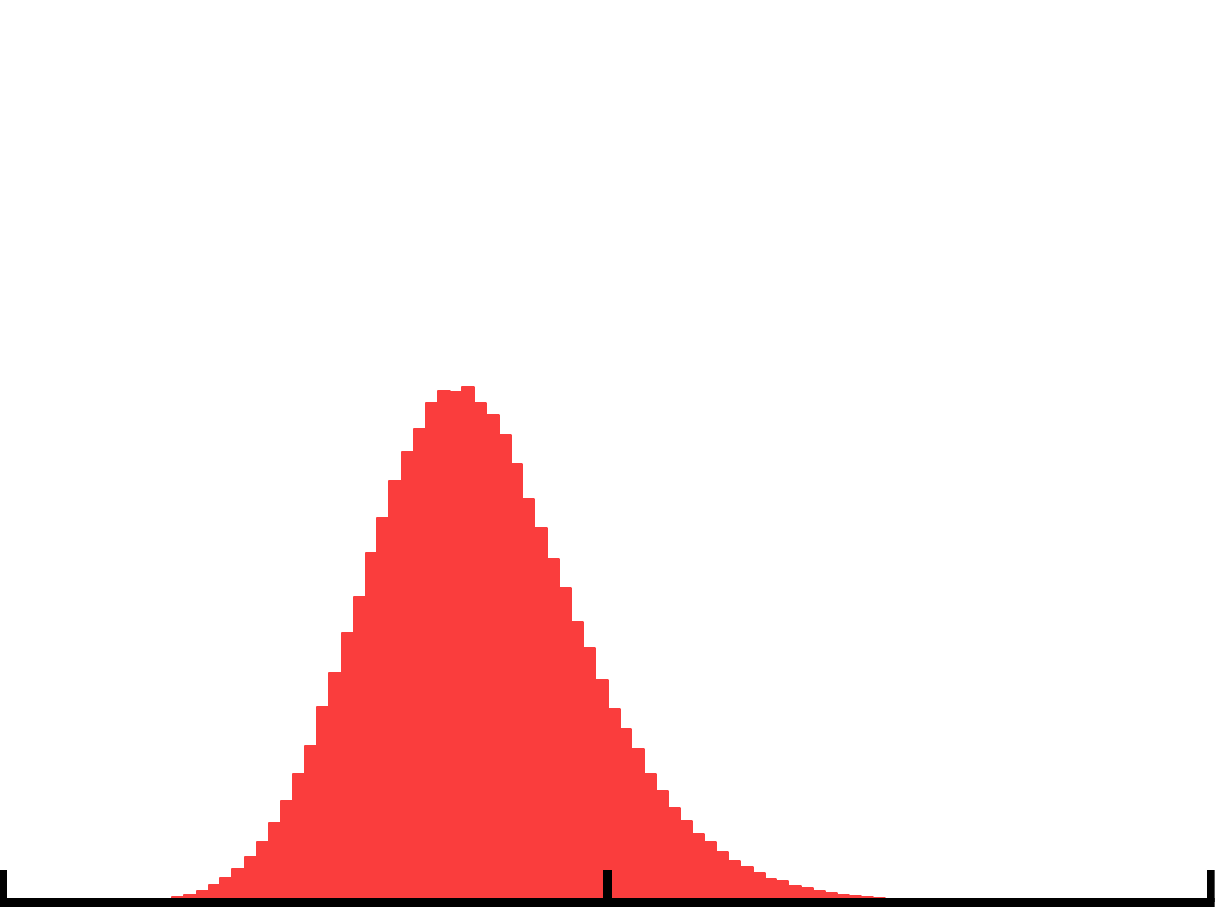}
\end{minipage}
\begin{minipage}[t]{0.24\linewidth}
\centering
\includegraphics[width=\linewidth]{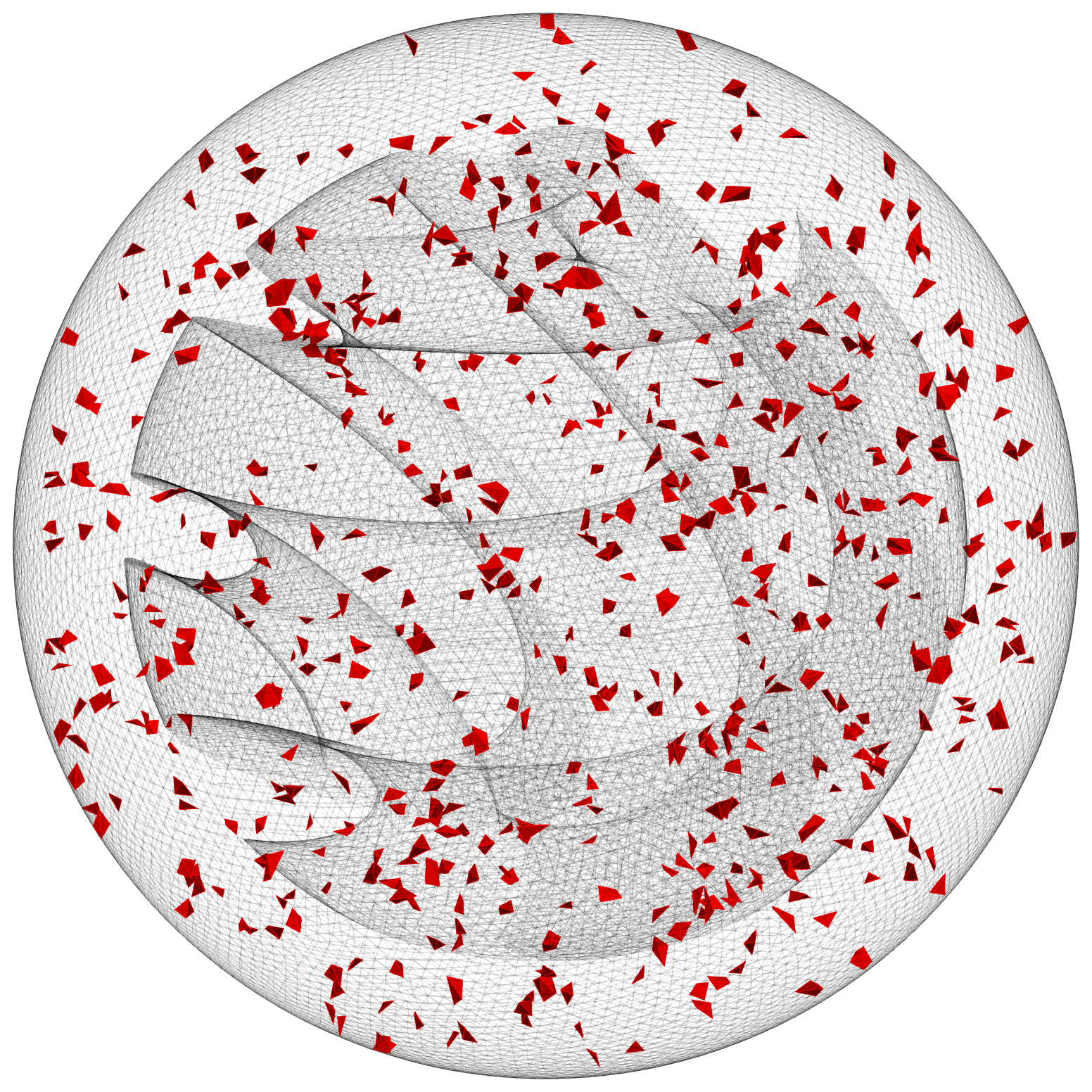}
\end{minipage}%
\begin{minipage}[t]{0.24\linewidth}
\centering
\includegraphics[width=\linewidth]{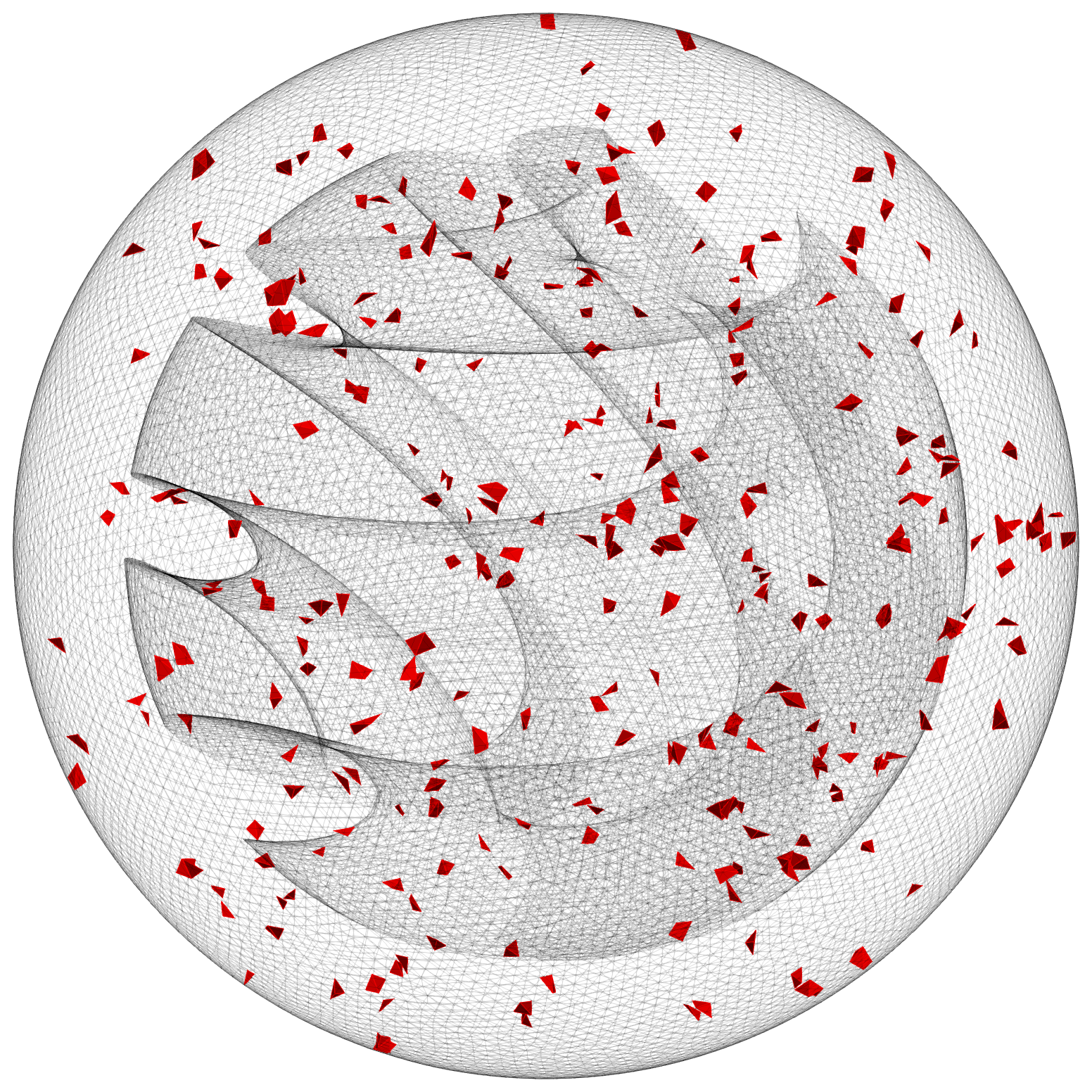}
\end{minipage}
\begin{minipage}[t]{0.24\linewidth}
\centering
\includegraphics[width=\linewidth]{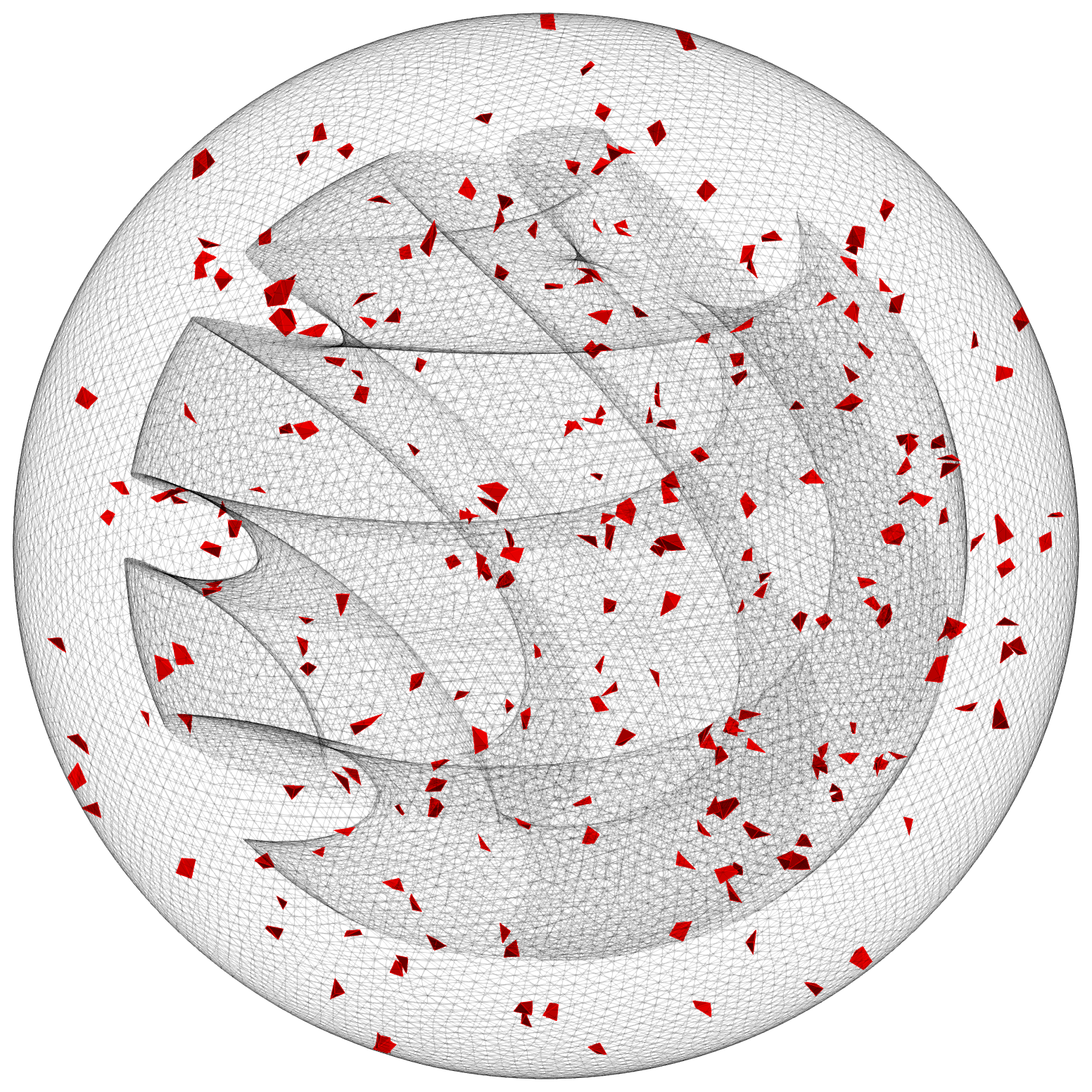}
\end{minipage}
\begin{minipage}[t]{0.24\linewidth}
\centering
\includegraphics[width=\linewidth]{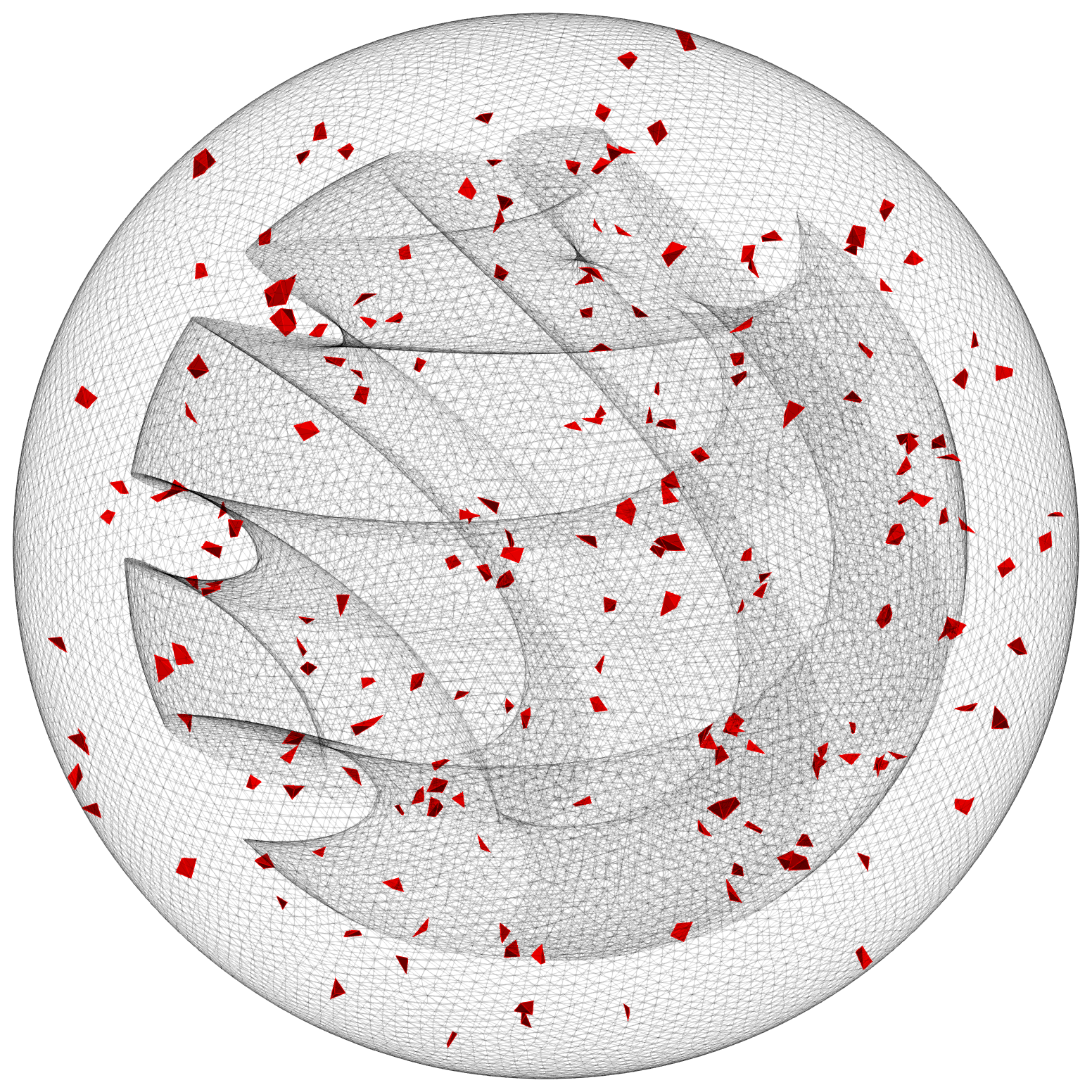}
\end{minipage}
\begin{minipage}[t]{0.24\linewidth}
\centering
\includegraphics[width=\linewidth]{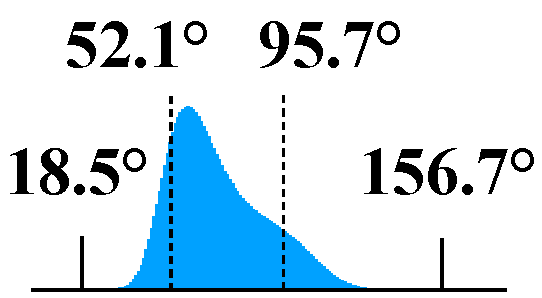}
\end{minipage}%
\begin{minipage}[t]{0.24\linewidth}
\centering
\includegraphics[width=\linewidth]{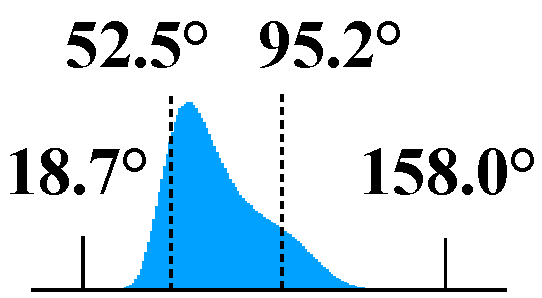}
\end{minipage}
\begin{minipage}[t]{0.24\linewidth}
\centering
\includegraphics[width=\linewidth]{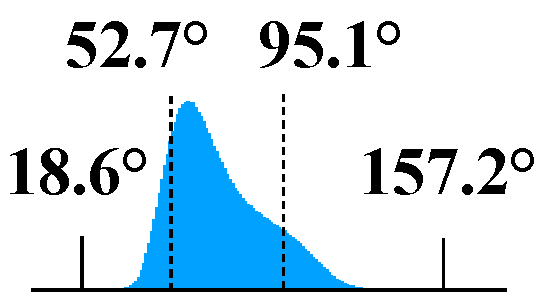}
\end{minipage}
\begin{minipage}[t]{0.24\linewidth}
\centering
\includegraphics[width=\linewidth]{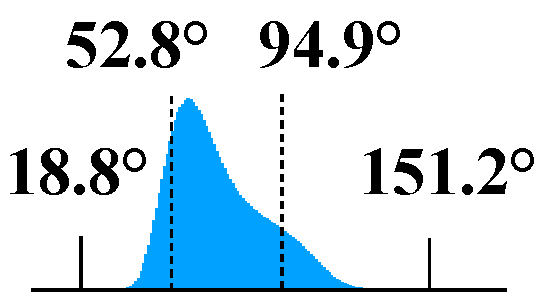}
\end{minipage}
\end{minipage}}%
\subcaptionbox{Output $\mathcal{T}$\label{fig:pipeline_output}}{%
\begin{minipage}[t]{0.16\linewidth}
\centering
\includegraphics[width=\linewidth]{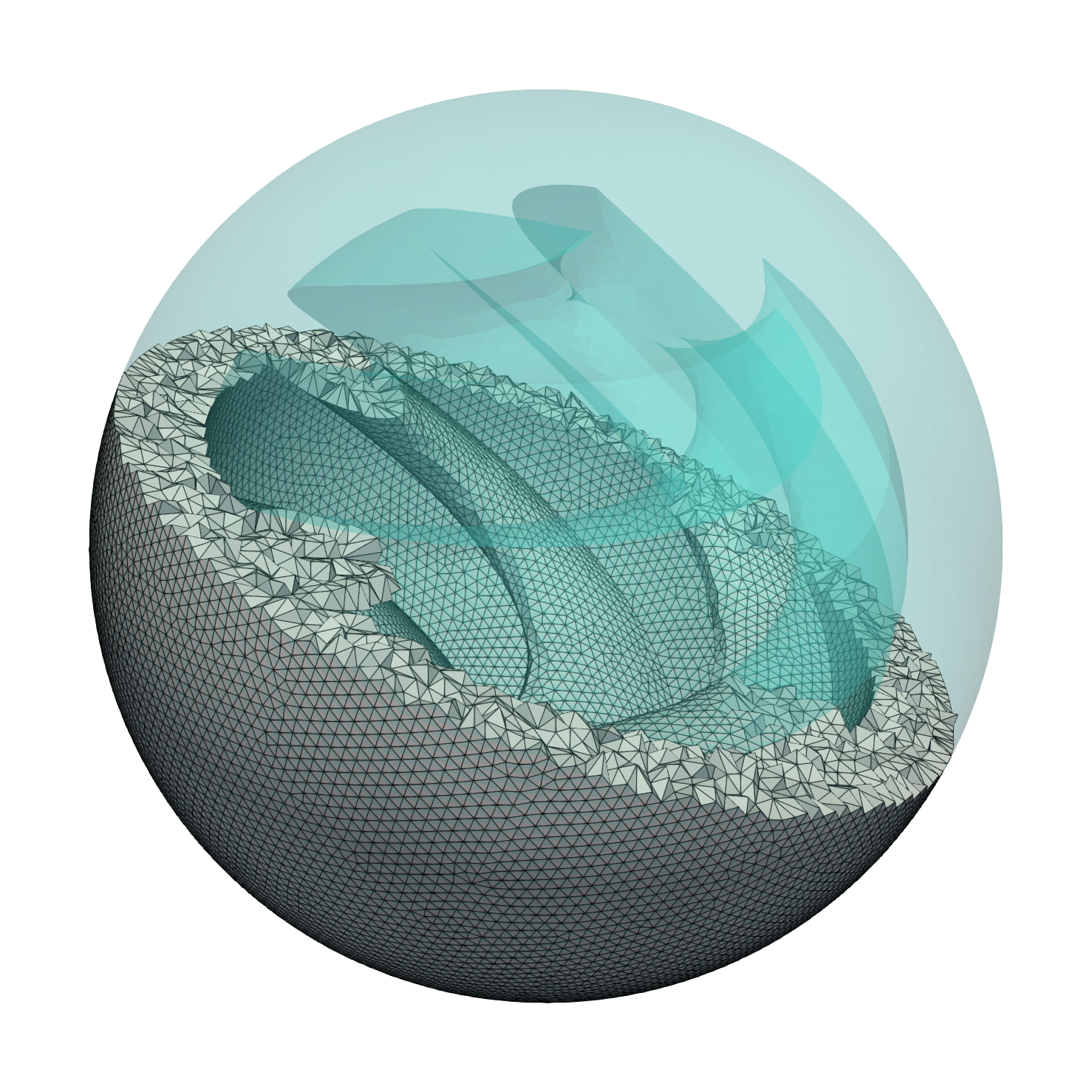}
\end{minipage}}%
\vspace{0.1in}
\caption{Algorithmic pipeline. (a) Our method takes a set of watertight, high-quality manifold triangular meshes as the input boundaries. (b) The initial tetrahedral mesh, constructed via Delaunay tetrahedralization, contains numerous slivers. Slivers with dihedral angles less than $30^\circ$ are rendered in red. Minimizing the volume-oriented energy improves the uniformity of the volume of tetrahedra. (c) Subsequently, minimizing the diheral angle-oriented energy further reduces the number of slivers. (d) The final tetrahedral mesh exhibits significantly improved mesh quality. The red and blue histograms represent the volume distribution and dihedral angle distribution of tetrahedral elements, respectively.  
In each angle histogram, we indicate the range of dihedral angles $[\theta_{\min}, \theta_{\max}]$ via short solid lines and the average of the minimal and maximal angles $\theta_{\min}^{\mathrm{avg}}$ and $\theta_{\max}^{\mathrm{avg}}$ via long dashed lines.
}
\label{fig:pipeline}
\end{figure*}

\section{Algorithm}\label{sec:Method}
\subsection{Overview}\label{subsec:overview}

Our algorithm takes a set of triangle meshes $\mathcal{M}$ as input. The exterior mesh represents the outer boundary of the object, while the remaining meshes define the boundaries of voids within the object. Since the boundary meshes serve as constraints in our method, we require that each triangle mesh is a closed manifold mesh with fairly regular triangulation. For meshes that do not meet this condition, preprocessing is required. 

Based on the user-specified target size $t$, we generate the initial tetrahedral mesh using a classical constrained Delaunay tetrahedralization method~\cite{cavalcanti1999three}, supplemented by the Bowyer-Watson vertex insertion algorithm~\cite{george1997improvements}, controlling the mesh edge lengths to be close to $t$. The initial tetrahedral mesh typically contains many slivers with very small dihedral angles.

Our method begins by minimizing the volume-oriented energy. Although this energy does not optimize individual tetrahedral shapes, it significantly improves the volume distribution of individual tetrahedra. Since slivers typically correspond to tetrahedra with extremely small volumes, reducing the volume-oriented energy effectively decreases the occurrence of slivers.

Following this initial volume-oriented optimization, WSVM proceeds by minimizing the dihedral angle-oriented energy, setting $\rho=1/A$.  This approach promotes uniform dihedral angles across the mesh, thereby significantly enhancing mesh quality.

Compared with the initial tetrahedral mesh, the output of WSVM is a high-quality tetrahedral mesh, characterized by a significant reduction in slivers and more uniform dihedral angles.
We show the pseudo-code of WSVM in Algorithm~\ref{algo:GlobalOptimization} and illustrate the algorithmic pipeline in Figure \ref{fig:pipeline}.

\subsection{Connectivity Update} 
\begin{figure}[htbp]
\centering
\includegraphics[width=1\linewidth]{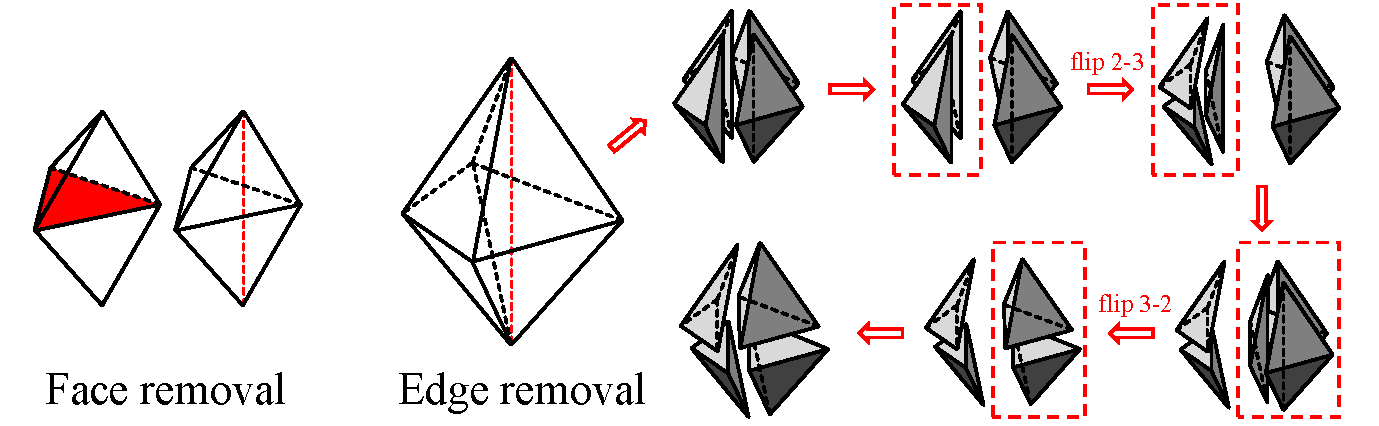}
\makebox[1.0in]{(a) 2-to-3 flip}
\makebox[2.2in]{(b) $n$-to-$m$ flip}
\caption{The flip operations for removing edges and faces.
(a) The 2-to-3 flip, used for removing the red face.
(b) A specific instance of an $n$-to-$m$ flip, where $n = m = 4$, is illustrated. This is achieved by applying a 2-to-3 flip followed by a 3-to-2 flip to remove the red edge.}
\label{fig:flipOperation}
\end{figure}

In addition to optimizing vertex positions, updating tetrahedral connectivity is also crucial. In our implementation, undesired edges and faces are removed via flip operations. As shown in Figure~\ref{fig:flipOperation}, face removal is accomplished through a 2-to-3 flip, and edge removal through an n-to-m flip~\cite{hang2015tetgen}. Given that three-dimensional flip operations are more complex and time-consuming than their two-dimensional counterparts, our method does not apply them indiscriminately. Instead, it selectively applies flips, assessing their effectiveness by the improvements they bring to the minimum dihedral angle, and focuses iteratively on elements with dihedral angles less than \(40^\circ\).


\subsection{Edge Split and Collapse}
Besides edge and face flipping, our algorithm also consists of edge splitting and collapsing to adjust the vertex count in the mesh. Let \( t \) be the user-specified threshold for edge length. In our implementation, edges longer than \( \frac{4}{3}t \) are split, while edges shorter than \( \frac{4}{5}t \) are collapsed. Additionally, our method prevents the collapse of edges if it would result in creating edges longer than \( \frac{4}{3}t \), a strategy that supports the convergence of the algorithm \cite{botsch2010polygon}.


\subsection{Numerical Solver}

\begin{algorithm}
\label{algo:NewtonMethodWithFixedIterations}
\begin{algorithmic}[1] 
\Function{SquaredVolumeMinimizing}{$\mathcal{T}, \rho$}
\Repeat
    \State Edge split and collapse
    \State Update connectivity
    \For{each internal vertex $\boldsymbol{v}_i\in\mathcal{T}$}
    \State Compute the energy $E\left(\rho_i,\boldsymbol{v}_i^{(0)}\right)$
      \State $k\leftarrow 0$
        \Repeat
        \State  Compute $\boldsymbol{g}^{(k)}$ and $\boldsymbol{H}^{(k)}$ of $E\left(\rho_i,\boldsymbol{v}_i^{(k)}\right)$             \State  $\boldsymbol{d}^{(k)} = - \left(\boldsymbol{H}^{(k)}\right)^{-1} \cdot \boldsymbol{g}^{(k)}$
             \State Find \( \alpha^{(k)} \) via line search along $\boldsymbol{d}^{(k)}$ 
            \State $\boldsymbol{v}_i^{(k+1)} = \boldsymbol{v}_i^{(k)} + \alpha^{(k)} \boldsymbol{d}^{(k)}$
           \State Compute the energy $E\left(\rho_i,\boldsymbol{v}_i^{(k+1)}\right)$
            \State $k\leftarrow k+1$
        
        \Until{$\left|E\left(\rho_i,\boldsymbol{v}_i^{(k+1)}\right) - E\left(\rho_i,\boldsymbol{v}_i^{(k)}\right)\right| \leq \epsilon_E$}
      \EndFor
\Until{the change of $\theta_{\min}^{\mathrm{avg}} \leq \epsilon_{\theta}$}
\State \textbf{return} ~$\mathcal{T}$
    \EndFunction
\end{algorithmic}
\end{algorithm}

We employ Newton's method to minimize both energies, using a backtracking line search \cite{backtracking} with constraints to efficiently minimize energy functional within a vertex's one-ring neighborhood and prevent negative volume tetrahedra. The method adjusts vertex positions based on the gradient \(\boldsymbol{g}^{(k)}\) and Hessian matrix \(\boldsymbol{H}^{(k)}\), and determines the step size \(\alpha^{(k)}\) to maintain mesh validity. 

\begin{table}[htbp]
  \centering
  \caption{Statistics. The best results are highlighted in blue.}
  \setlength{\tabcolsep}{1.5pt} 
  \begin{scriptsize}
    \begin{tabular}{cccccccccccccc}
    \hline
    \multirow{2}[3]{*}{{Model}} & \multirow{2}[3]{*}{{Method}} & \multicolumn{4}{c|}{$\theta$}     & \multicolumn{2}{c|}{$c$} & \multicolumn{2}{c|}{$e$} & \multicolumn{2}{c}{$s$} & \multirow{2}[3]{*}{ $N_t/k$} & \multirow{2}[3]{*}{{$P$/\%}} \\
\cmidrule{3-12}          & \multicolumn{1}{c}{} & \multicolumn{1}{c}{$\theta_{\min}^{\mathrm{avg}}$} & \multicolumn{1}{c}{$\theta_{\min}$} & \multicolumn{1}{c}{$\theta_{\max}^{\mathrm{avg}}$} & \multicolumn{1}{c}{$\theta_{\max}$} & \multicolumn{1}{c}{{Avg.}} & \multicolumn{1}{c}{{Max.}} & \multicolumn{1}{c}{{Avg.}} & \multicolumn{1}{c}{{Max.}} & \multicolumn{1}{c}{{Avg.}} & \multicolumn{1}{c}{{Max.}} &       
&  \\
\cmidrule{1-14}
\cmidrule{1-12}\cmidrule{14-14}\multirow{5}[1]{*}{\rotatebox{90}{Botijo}} & WSVM  & \cellcolor[rgb]{ .678,  .847,  .902}52.4 & \cellcolor[rgb]{ .678,  .847,  .902}20.2 & \cellcolor[rgb]{ .678,  .847,  .902}95.2 & \cellcolor[rgb]{ .678,  .847,  .902}155.7 & \cellcolor[rgb]{ .678,  .847,  .902}1.13 & \cellcolor[rgb]{ .678,  .847,  .902}2.7 & \cellcolor[rgb]{ .678,  .847,  .902}1.42 & 2.75  & \cellcolor[rgb]{ .678,  .847,  .902}0.27 & 0.65  & 197   & \cellcolor[rgb]{ .678,  .847,  .902}0.05  \\
      & Gmsh  & 47.6  & 4.7   & 101.7 & 172.5 & 1.26  & 9.89  & 1.52  & \cellcolor[rgb]{ .678,  .847,  .902}2.41 & 0.32  & \cellcolor[rgb]{ .678,  .847,  .902}0.59 & 208   & 6.67  \\
      & MMG3D   & 47.4  & 12.8  & 98.2  & 157.3 & 1.2   & 4.23  & 1.6   & 8.28  & 0.35  & 0.85  & 191   & 1.38  \\
      & TetGen & 49.5  & 1.6   & 96.9  & 176.6 & 1.2   & 25.66 & 1.47  & 2.7   & 0.29  & 0.64  & 216   & 4.35  \\
      & fTetwild & 50    & 13.7  & 96.3  & 161.4 & 1.17  & 3.73  & 1.55  & 3.41  & 0.32  & 0.61  & 185   & 0.21  \\
      \hline
\multirow{5}[0]{*}{\rotatebox{90}{Bunny}} & WSVM  & \cellcolor[rgb]{ .678,  .847,  .902}52.6 & \cellcolor[rgb]{ .678,  .847,  .902}22.4 & \cellcolor[rgb]{ .678,  .847,  .902}94.9 & \cellcolor[rgb]{ .678,  .847,  .902}152.5 & \cellcolor[rgb]{ .678,  .847,  .902}1.12 & \cellcolor[rgb]{ .678,  .847,  .902}2.5 & \cellcolor[rgb]{ .678,  .847,  .902}1.42 & 2.99  & \cellcolor[rgb]{ .678,  .847,  .902}0.26 & 0.66  & 284   & \cellcolor[rgb]{ .678,  .847,  .902}0.03  \\
      & Gmsh  & 47.5  & 12.8  & 101.8 & 156.8 & 1.26  & 3.37  & 1.52  & \cellcolor[rgb]{ .678,  .847,  .902}2.78 & 0.32  & 0.65  & 299   & 6.76  \\
      & MMG3D   & 47.3  & 12    & 98.3  & 156.9 & 1.2   & 3.97  & 1.6   & 6.14  & 0.35  & 0.82  & 269   & 1.33  \\
      & TetGen & 49.6  & 1.7   & 96.7  & 176.9 & 1.2   & 26.77 & 1.47  & 3.03  & 0.29  & 0.65  & 315   & 4.22  \\
      & fTetwild & 50    & 13.9  & 96.3  & 159.3 & 1.17  & 3.58  & 1.55  & 3.21  & 0.32  & \cellcolor[rgb]{ .678,  .847,  .902}0.64 & 258   & 0.18  \\
      \hline
\multirow{5}[0]{*}{\rotatebox{90}{Carter}} & WSVM  & \cellcolor[rgb]{ .678,  .847,  .902}51.4 & 16.5  & 96.1  & 160.2 & \cellcolor[rgb]{ .678,  .847,  .902}1.14 & 3.89  & \cellcolor[rgb]{ .678,  .847,  .902}1.46 & 4.44  & \cellcolor[rgb]{ .678,  .847,  .902}0.28 & 0.78  & 463   & 0.13  \\
      & Gmsh  & 47.5  & 5.1   & 101.7 & 171.4 & 1.26  & 9.4   & 1.53  & 3.48  & 0.32  & 0.72  & 463   & 6.65  \\
      & MMG3D   & 46.9  & 13.4  & 98.8  & 156.1 & 1.21  & 4.05  & 1.63  & 4.15  & 0.36  & 0.8   & 500   & 1.54  \\
      & TetGen & 49.5  & 1.4   & 97.1  & 176.8 & 1.2   & 29.62 & 1.48  & \cellcolor[rgb]{ .678,  .847,  .902}3.33 & 0.29  & \cellcolor[rgb]{ .678,  .847,  .902}0.7 & 502   & 4.26  \\
      & fTetwild & 51.2  & \cellcolor[rgb]{ .678,  .847,  .902}16.9 & \cellcolor[rgb]{ .678,  .847,  .902}94.2 & \cellcolor[rgb]{ .678,  .847,  .902}147.4 & \cellcolor[rgb]{ .678,  .847,  .902}1.14 & \cellcolor[rgb]{ .678,  .847,  .902}2.95 & 1.49  & 3.65  & 0.29  & 0.77  & 460   & \cellcolor[rgb]{ .678,  .847,  .902}0.07  \\
      \hline
\multirow{5}[0]{*}{\rotatebox{90}{Lion}} & WSVM  & \cellcolor[rgb]{ .678,  .847,  .902}52.4 & \cellcolor[rgb]{ .678,  .847,  .902}16.3 & \cellcolor[rgb]{ .678,  .847,  .902}95 & 158   & \cellcolor[rgb]{ .678,  .847,  .902}1.12 & \cellcolor[rgb]{ .678,  .847,  .902}3.16 & \cellcolor[rgb]{ .678,  .847,  .902}1.42 & 3     & \cellcolor[rgb]{ .678,  .847,  .902}0.26 & 0.69  & 295   & \cellcolor[rgb]{ .678,  .847,  .902}0.05  \\
      & Gmsh  & 47.5  & 12.3  & 101.7 & \cellcolor[rgb]{ .678,  .847,  .902}156.8 & 1.26  & 3.37  & 1.52  & \cellcolor[rgb]{ .678,  .847,  .902}2.73 & 0.32  & \cellcolor[rgb]{ .678,  .847,  .902}0.65 & 307   & 6.70  \\
      & MMG3D   & 47.3  & 12.3  & 98.3  & 156.9 & 1.2   & 4.28  & 1.61  & 4.97  & 0.35  & 0.81  & 277   & 1.36  \\
      & TetGen & 49.6  & 1.7   & 96.7  & 177   & 1.2   & 26.79 & 1.47  & 2.81  & 0.29  & 0.66  & 325   & 4.24  \\
      & fTetwild & 50    & 15    & 95.9  & 162.9 & 1.17  & 3.86  & 1.55  & 2.91  & 0.32  & 0.67  & 275   & 0.19  \\
      \hline
\multirow{5}[0]{*}{\rotatebox{90}{Elephant}} & WSVM  & \cellcolor[rgb]{ .678,  .847,  .902}52.4 & \cellcolor[rgb]{ .678,  .847,  .902}22.8 & \cellcolor[rgb]{ .678,  .847,  .902}95.1 & \cellcolor[rgb]{ .678,  .847,  .902}149.4 & \cellcolor[rgb]{ .678,  .847,  .902}1.12 & \cellcolor[rgb]{ .678,  .847,  .902}2.32 & \cellcolor[rgb]{ .678,  .847,  .902}1.42 & 2.78  & \cellcolor[rgb]{ .678,  .847,  .902}0.26 & 0.66  & 359   & \cellcolor[rgb]{ .678,  .847,  .902}0.04  \\
      & Gmsh  & 47.5  & 12.8  & 101.8 & 158.8 & 1.26  & 3.74  & 1.52  & \cellcolor[rgb]{ .678,  .847,  .902}2.44 & 0.32  & \cellcolor[rgb]{ .678,  .847,  .902}0.61 & 374   & 6.70  \\
      & MMG3D   & 47.2  & 13.1  & 98.4  & 156   & 1.2   & 4.02  & 1.61  & 4.29  & 0.35  & 0.81  & 370   & 1.29  \\
      & TetGen & 49.5  & 1.9   & 96.8  & 176.9 & 1.2   & 25.75 & 1.47  & 2.95  & 0.29  & 0.68  & 395   & 4.24  \\
      & fTetwild & 50.8  & 18.5  & 95.6  & 156.6 & 1.15  & 2.83  & 1.5   & 2.86  & 0.3   & 0.67  & 352   & 0.20  \\
      \hline
\multirow{5}[0]{*}{\rotatebox{90}{Fandisk}} & WSVM  & \cellcolor[rgb]{ .678,  .847,  .902}52.3 & 17    & \cellcolor[rgb]{ .678,  .847,  .902}95.3 & 154.5 & \cellcolor[rgb]{ .678,  .847,  .902}1.13 & 2.62  & \cellcolor[rgb]{ .678,  .847,  .902}1.43 & 3.79  & \cellcolor[rgb]{ .678,  .847,  .902}0.27 & 0.7   & 268   & \cellcolor[rgb]{ .678,  .847,  .902}0.07  \\
      & Gmsh  & 47.5  & 10.4  & 101.8 & 163.8 & 1.26  & 4.95  & 1.52  & \cellcolor[rgb]{ .678,  .847,  .902}2.96 & 0.32  & \cellcolor[rgb]{ .678,  .847,  .902}0.68 & 283   & 6.59  \\
      & MMG3D   & 47.1  & 15.7  & 98.6  & \cellcolor[rgb]{ .678,  .847,  .902}148.9 & 1.21  & 3.25  & 1.62  & 3.88  & 0.35  & 0.77  & 294   & 1.37  \\
      & TetGen & 49.6  & 1.8   & 96.8  & 176.9 & 1.2   & 25.46 & 1.47  & 3     & 0.29  & \cellcolor[rgb]{ .678,  .847,  .902}0.68 & 297   & 4.17  \\
      & fTetwild & 49.6  & \cellcolor[rgb]{ .678,  .847,  .902}20.2 & 96.6  & 150.6 & 1.17  & \cellcolor[rgb]{ .678,  .847,  .902}2.43 & 1.56  & 3.47  & 0.32  & \cellcolor[rgb]{ .678,  .847,  .902}0.68 & 266   & 0.18  \\
      \hline
\multirow{5}[0]{*}{\rotatebox{90}{Fertility}} & WSVM  & \cellcolor[rgb]{ .678,  .847,  .902}52.3 & \cellcolor[rgb]{ .678,  .847,  .902}19.5 & \cellcolor[rgb]{ .678,  .847,  .902}95.3 & 153.3 & \cellcolor[rgb]{ .678,  .847,  .902}1.13 & 2.68  & \cellcolor[rgb]{ .678,  .847,  .902}1.43 & 3.62  & \cellcolor[rgb]{ .678,  .847,  .902}0.27 & 0.7   & 321   & \cellcolor[rgb]{ .678,  .847,  .902}0.05  \\
      & Gmsh  & 47.5  & 7.1   & 101.9 & 167.8 & 1.26  & 7.28  & 1.52  & 3.55  & 0.32  & 0.73  & 330   & 6.82  \\
      & MMG3D   & 47    & 14.7  & 98.6  & 155   & 1.21  & 3.49  & 1.62  & 4.74  & 0.35  & 0.79  & 343   & 1.45  \\
      & TetGen & 49.5  & 1.7   & 96.9  & 176.9 & 1.2   & 25.61 & 1.47  & 3.02  & 0.29  & \cellcolor[rgb]{ .678,  .847,  .902}0.68 & 351   & 4.29  \\
      & fTetwild & 49.7  & 18.7  & 96.3  & \cellcolor[rgb]{ .678,  .847,  .902}151.7 & 1.17  & \cellcolor[rgb]{ .678,  .847,  .902}2.66 & 1.54  & \cellcolor[rgb]{ .678,  .847,  .902}2.99 & 0.32  & 0.7   & 306   & 0.25  \\
      \hline
\multirow{5}[0]{*}{\rotatebox{90}{Molecule}} & WSVM  & \cellcolor[rgb]{ .678,  .847,  .902}50.6 & \cellcolor[rgb]{ .678,  .847,  .902}17.6 & 96.2  & \cellcolor[rgb]{ .678,  .847,  .902}156.6 & \cellcolor[rgb]{ .678,  .847,  .902}1.15 & \cellcolor[rgb]{ .678,  .847,  .902}3.07 & \cellcolor[rgb]{ .678,  .847,  .902}1.49 & 3.2   & \cellcolor[rgb]{ .678,  .847,  .902}0.3 & 0.68  & 99    & \cellcolor[rgb]{ .678,  .847,  .902}0.14  \\
      & Gmsh  & 47.5  & 5.5   & 100.8 & 172.8 & 1.26  & 9.75  & 1.56  & 2.76  & 0.33  & 0.67  & 92    & 7.11  \\
      & MMG3D   & 46.7  & 13.5  & 99    & 158.1 & 1.22  & 4.18  & 1.64  & 6.83  & 0.36  & 0.83  & 103   & 2.29  \\
      & TetGen & 48.9  & 1.9   & 97.6  & 176.7 & 1.22  & 24.24 & \cellcolor[rgb]{ .678,  .847,  .902}1.49 & \cellcolor[rgb]{ .678,  .847,  .902}2.6 & \cellcolor[rgb]{ .678,  .847,  .902}0.3 & \cellcolor[rgb]{ .678,  .847,  .902}0.62 & 104   & 4.93  \\
      & fTetwild & 49.7  & 8.6   & \cellcolor[rgb]{ .678,  .847,  .902}95.7 & 162.7 & 1.16  & 4.75  & 1.54  & 2.69  & 0.32  & 0.65  & 105   & 0.24  \\
      \hline
\multirow{5}[0]{*}{\rotatebox{90}{Pegaso}} & WSVM  & \cellcolor[rgb]{ .678,  .847,  .902}52.2 & \cellcolor[rgb]{ .678,  .847,  .902}18.4 & \cellcolor[rgb]{ .678,  .847,  .902}95.4 & 154.3 & \cellcolor[rgb]{ .678,  .847,  .902}1.13 & \cellcolor[rgb]{ .678,  .847,  .902}2.71 & \cellcolor[rgb]{ .678,  .847,  .902}1.43 & 2.78  & \cellcolor[rgb]{ .678,  .847,  .902}0.27 & 0.69  & 420   & \cellcolor[rgb]{ .678,  .847,  .902}0.05  \\
      & Gmsh  & 47.6  & 10.7  & 101.8 & 159.9 & 1.26  & 4.07  & 1.52  & \cellcolor[rgb]{ .678,  .847,  .902}2.51 & 0.32  & \cellcolor[rgb]{ .678,  .847,  .902}0.59 & 430   & 6.64  \\
      & MMG3D   & 47.1  & 12    & 98.6  & \cellcolor[rgb]{ .678,  .847,  .902}153.7 & 1.21  & 3.76  & 1.62  & 4.88  & 0.35  & 0.79  & 459   & 1.43  \\
      & TetGen & 49.5  & 1.7   & 97    & 177   & 1.2   & 26.86 & 1.47  & 3.05  & 0.29  & 0.69  & 455   & 4.24  \\
      & fTetwild & 49.7  & 15.4  & 96.2  & 155   & 1.16  & 3.02  & 1.54  & 3.2   & 0.32  & 0.67  & 407   & 0.25  \\
      \hline
\multirow{5}[0]{*}{\rotatebox{90}{Sphere}} & WSVM  & \cellcolor[rgb]{ .678,  .847,  .902}52.6 & 17.8  & \cellcolor[rgb]{ .678,  .847,  .902}95.1 & 152.6 & \cellcolor[rgb]{ .678,  .847,  .902}1.12 & 2.87  & \cellcolor[rgb]{ .678,  .847,  .902}1.42 & 2.38  & \cellcolor[rgb]{ .678,  .847,  .902}0.26 & \cellcolor[rgb]{ .678,  .847,  .902}0.59 & 55    & \cellcolor[rgb]{ .678,  .847,  .902}0.03  \\
      & Gmsh  & 47.6  & 13    & 101.9 & 156.6 & 1.26  & 3.34  & 1.51  & \cellcolor[rgb]{ .678,  .847,  .902}2.29 & 0.31  & \cellcolor[rgb]{ .678,  .847,  .902}0.59 & 56    & 6.53  \\
      & MMG3D   & 47.3  & 13.8  & 98.2  & \cellcolor[rgb]{ .678,  .847,  .902}150.2 & 1.2   & 3.56  & 1.6   & 3.29  & 0.35  & 0.74  & 52    & 1.54  \\
      & TetGen & 49.6  & 2.8   & 96.9  & 175.8 & 1.2   & 18.11 & 1.47  & 2.72  & 0.29  & 0.64  & 60    & 4.28  \\
      & fTetwild & 52.2  & \cellcolor[rgb]{ .678,  .847,  .902}18 & 95.3  & 151.5 & 1.14  & \cellcolor[rgb]{ .678,  .847,  .902}2.62 & 1.45  & 2.75  & \cellcolor[rgb]{ .678,  .847,  .902}0.26 & 0.7   & 39    & 0.71  \\
          \hline
    \end{tabular}%
  \label{tab:MeshQuality}%
  \end{scriptsize}
\end{table}%

\begin{figure*}[htbp]
\centering
\includegraphics[width=1\linewidth]{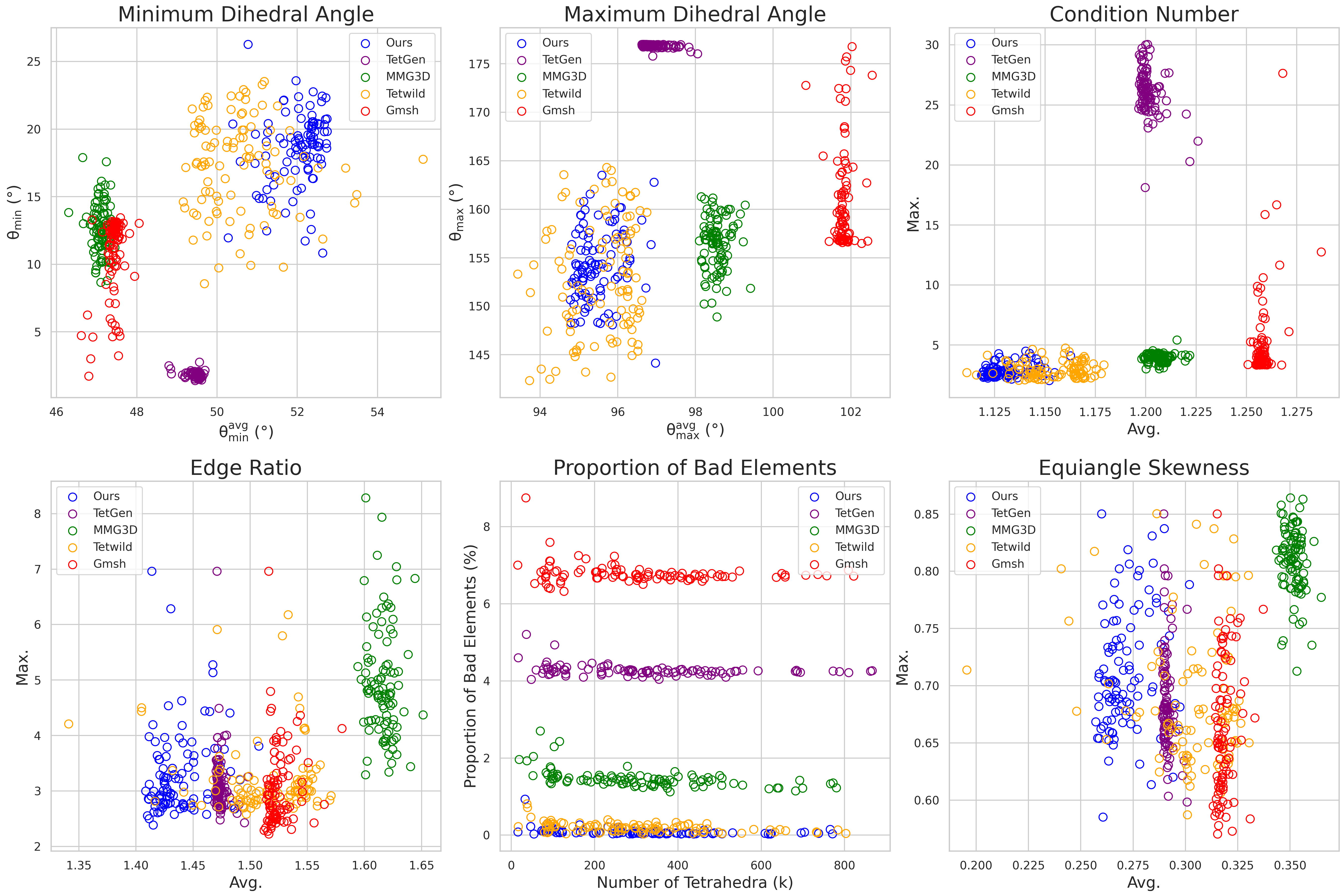}
\caption{Comparison of WSVM with various methods across different metrics. In the scatter plots, data points in the top-right corner represent better performance for the minimum dihedral angle. Data points in the bottom-left corner indicate better performance for the maximum dihedral angle, condition number, edge ratio, and equiangle skewness. Lower data points in the proportion of bad elements plot suggest a smaller percentage of poor-quality elements.}
\label{fig:result_compare}
\end{figure*}

\section{Experiments}\label{sec:experiment}

We implemented our algorithm in C++ and tested on a laptop with 32 GB of memory, equipped with an AMD Ryzen 9 7945HX CPU.

\paragraph{Baselines}
We compared our method against four representative methods: TetGen \cite{hang2015tetgen}, Gmsh \cite{GmshPaper}, MMG3D \cite{DAPOGNY2014358}, and fTetWild \cite{hu2020fasttetwild}. To ensure a fair comparison, we carefully tuned parameters so that the meshes generated for each test model had a similar number of tetrahedra. Specifically, for TetGen, we adjusted its parameters to align the average volume of its tetrahedral elements with ours, thereby controlling the tetrahedron count. For Gmsh and MMG3D, we controlled the mesh size to match our mesh scale. For fTetWild, we used the input target edge length to set the mesh scale. All other parameters were maintained at their default values. As reported in Table \ref{tab:MeshQuality}, the number of tetrahedral elements \( N_t \) generated by each method is similar, facilitating a fair comparison.

\paragraph{Dataset}
We selected 100 models, most of which are from the dataset provided by \cite{yang2020error} and \cite{levy2010p}. Before testing, each model underwent isotropic remeshing to obtain a high-quality surface mesh. The dataset we tested covers a wide range of geometric and topological features, including high-genus models, models with sharp features, as well as CAD models commonly used in industrial applications.

\begin{figure}[htbp]
\centering
\includegraphics[width=\linewidth]{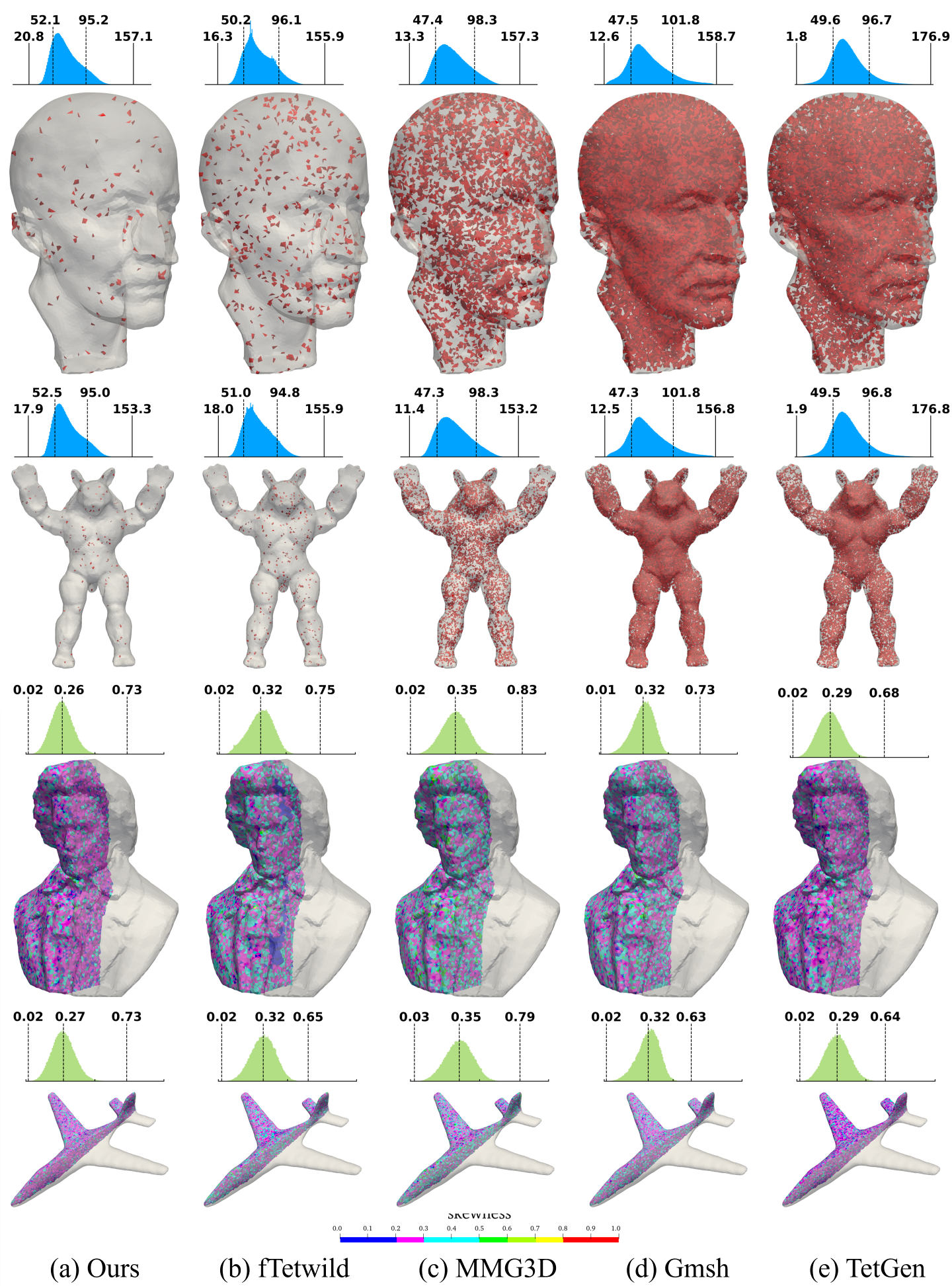} 
\caption{Comparison of WSVM with various methods, for the first two models, elements with dihedral angles less than 30 degrees are rendered in red. For the latter two models, we show the mesh visualization colored by equiangle skewness values, along with corresponding histograms indicating the distribution of skewness values. The histograms highlight the minimum, average, and maximum skewness values, represented by green bars on the plot.}
\label{fig:compare_method}
\end{figure}


\paragraph{Metrics}
We assess the quality of tetrahedral meshes using five commonly used metrics: 
1) \textbf{Distribution of dihedral angles $\theta$}, which characterizes the regularity of tetrahedral elements. We report the minimum dihedral angle $\theta_{\min}$, the average minimum $\theta_{\min}^{\text{avg}}$, the maximum dihedral angle $\theta_{\max}$, and the average maximum $\theta_{\max}^{\text{avg}}$.
2) \textbf{Edge ratio} $e \in [1, +\infty)$, defined as the ratio of the longest edge to the shortest edge of each tetrahedron \cite{AspectRatio}. We evaluate the average and maximum edge ratios among all tetrahedra, with the optimal value being one.
3) \textbf{Condition number} $c \in [1, +\infty)$, which measures the deviation of an element from a regular tetrahedron~\cite{Condition}. The optimal value is one.
4) \textbf{Equiangle skewness} $s \in [0, 1]$, which measures the maximum ratio of an element's angle to the angle of an equilateral element\footnote{\url{https://www.pointwise.com/doc/user-manual/examine/functions/equiangle-skewness.html}}. The optimal value is zero.
5) \textbf{Proportion of bad elements} $P \in [0, 100]$. Tetrahedra with dihedral angles less than 30° are considered bad elements, and we calculate their percentage among the total tetrahedra. A lower percentage indicates higher uniformity and angle quality.

\paragraph{Performance}  
We validated five methods on the dataset. Due to the high-quality manifold meshes used, our primary focus is on the quality of the final generated tetrahedral meshes.
In Figure \ref{fig:result_compare}, we use scatter plots to compare the mesh quality metrics of all methods applied to the dataset, providing an intuitive comparison. 
In Table \ref{tab:MeshQuality}, we present mesh quality data for 10 selected models. 
The results demonstrate that our method excels in controlling both small and large angles in the tetrahedral meshes. It also performs well in terms of condition number, edge ratio, and equiangle skewness, maintaining poor elements within a minimal range.
Additionally, we selected 20 models from the dataset, covering common graphics test models, CAD models, and high-genus models, and visualized them as shown in Figure \ref{fig:gallery}.


\begin{figure}[htbp]
\centering
\includegraphics[width=\linewidth]{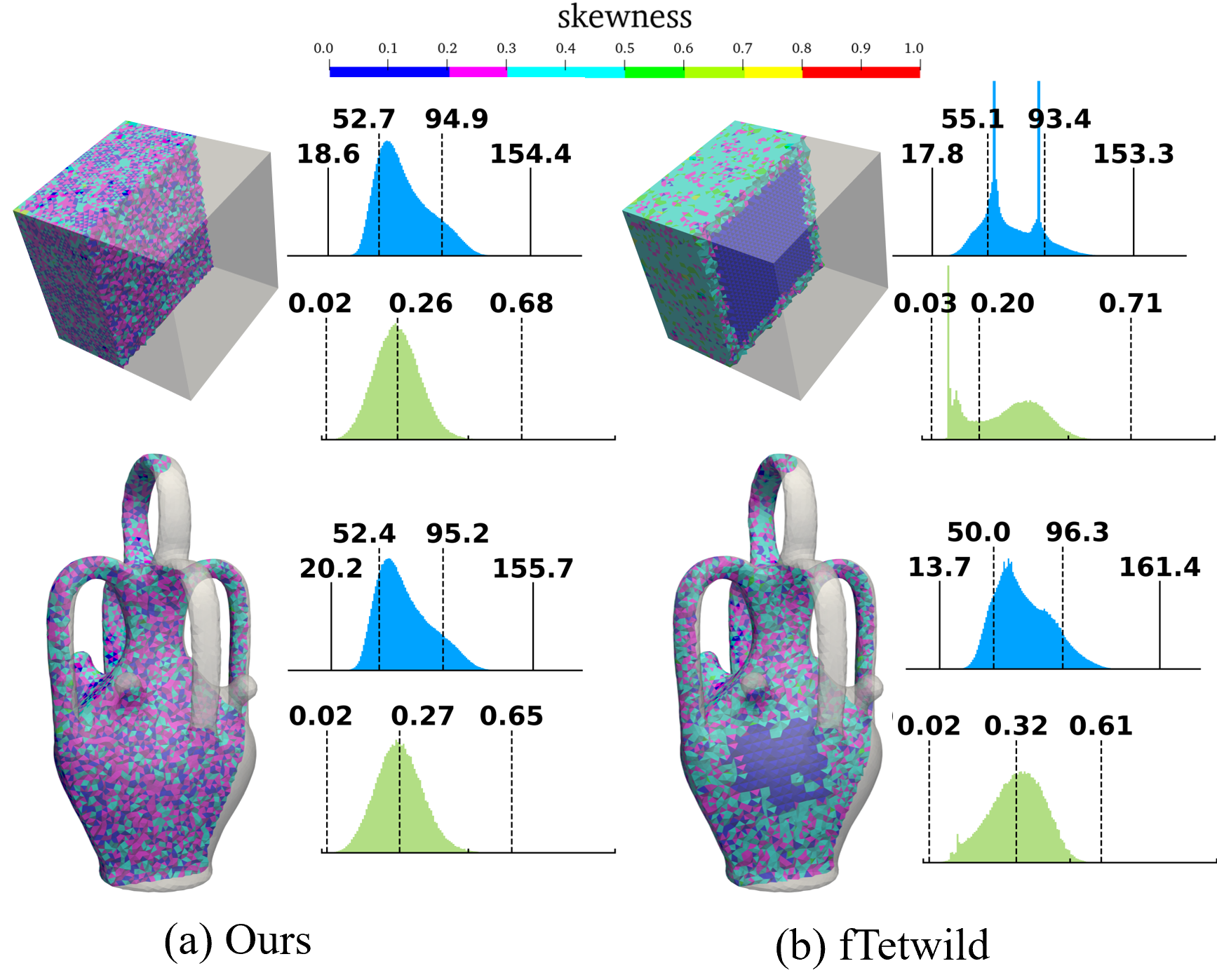}
\caption{The comparison of meshes generated by WSVM and fTetWild was performed using the simple Cube model and the high-genus Botijo model, focusing on dihedral angle and equiangle skewness metrics.}
\label{fig:tetwild_compare}
\end{figure}

\paragraph{Comparison with SOTA software}
TetGen and Gmsh are widely-used tetrahedral mesh generation programs based on Delaunay triangulation, employing highly efficient vertex insertion algorithms. As shown in Figure \ref{fig:result_compare}, our results indicate that TetGen performs well in terms of edge ratio, while Gmsh excels at controlling the worst equiangle skewness.
MMG3D is an popular tetrahedral remeshing tool that operates relatively slower than the previous two programs. However, it demonstrates superior control over maximum and minimum angles, as well as a better proportion of poor-quality elements.

To provide a more intuitive comparison of mesh quality, we selected several models and compared the poor-quality elements and overall mesh quality generated by each method, as shown in Figure \ref{fig:compare_method}. Our approach and fTetWild generate the highest mesh quality among all methods. However, our approach and fTetWild are also the slowest, typically operating an order of magnitude slower than the other methods. This is primarily because both employ extensive remeshing operations to improve mesh quality, which often results in a trade-off between high-quality tetrahedra and computational efficiency.

FTetWild is a robust tetrahedralization algorithm known for its high tolerance to imperfect inputs, achieved by allowing modifications to the input boundary mesh within a certain tolerance range. 
The algorithm uses a voxel-based point insertion strategy and optimizes vertices by minimizing AMIPS energy~\cite{AMIPSrabinovich2017scalable}. FTetWild generates right-angled tetrahedra, effectively preventing the creation of slivers in the interior. This strategy is particularly effective for models with simple or axis-aligned boundaries, as demonstrated by the Cube model in Figure \ref{fig:tetwild_compare}. However, this approach tends to produce non-uniform dihedral angles, typically peaking around $60^{\circ}$ and $90^{\circ}$, due to the prevalence of right-angled interior tetrahedra. Additionally, when dealing with high-genus models with non-axis-aligned boundaries, such as the Botijo model, this approach faces difficulty in producing high-quality tetrahedral elements near the boundary. 

In contrast to fTetWild, which is designed for robustness and handling imperfect inputs, WSVM assumes the input mesh is of high quality. As shown in the Figure \ref{fig:result_compare}, our method demonstrates advantages in terms of condition number, edge ratio, and equiangle skewness. WSVM aims to produce high-quality tetrahedral elements with fewer slivers and greater uniformity in tetrahedral shapes. Therefore, we consider WSVM and fTetWild to be complementary, each serving distinct application domains.

\begin{figure}[htbp]
\centering
\includegraphics[width=0.8\linewidth]{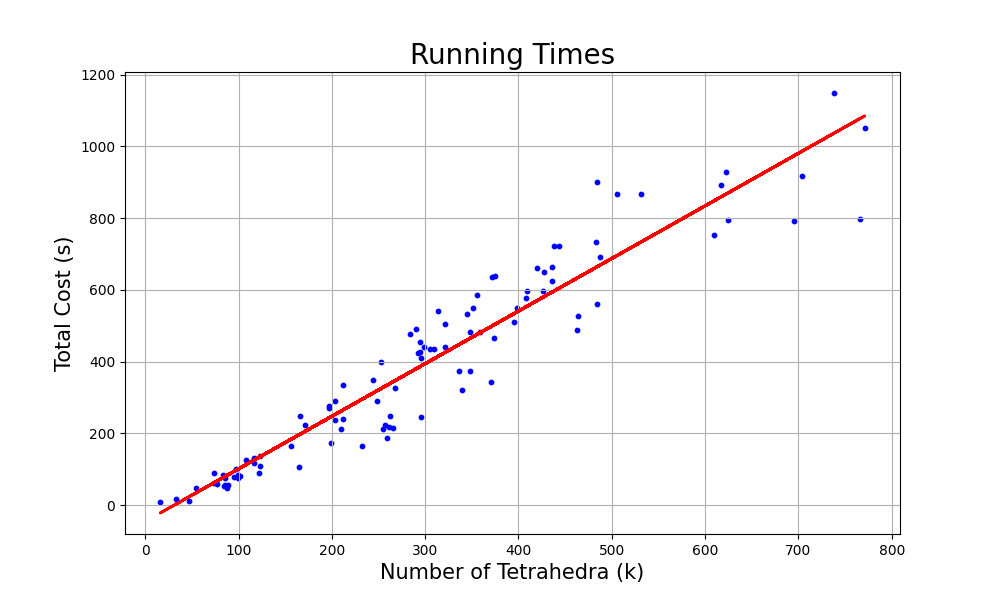}
\caption{WSVM exhibits a roughly linear time complexity with respect to the number of tetrahedral elements in the output mesh.}
\label{fig:timedata}
\end{figure}

\paragraph{Running times} We analyzed the running times by plotting scatter points of mesh sizes against running times for all models and performing a least squares linear fit, as shown in Figure \ref{fig:timedata}. 
Due to its convex nature, optimizing the weighted square volume energies for each vertex is highly efficient. Based on a high-quality surface mesh and initial Delaunay tetrahedralization, global convergence typically requires no more than one hundred iterations.
However, the main bottleneck is the tetrahedral flip operations, which are more complex and time-consuming compared to their two-dimensional counterparts. Despite of this, our algorithm maintains running times within the minute range for all tested models.

\begin{figure}[htbp]
\centering
\includegraphics[width=0.8\linewidth]{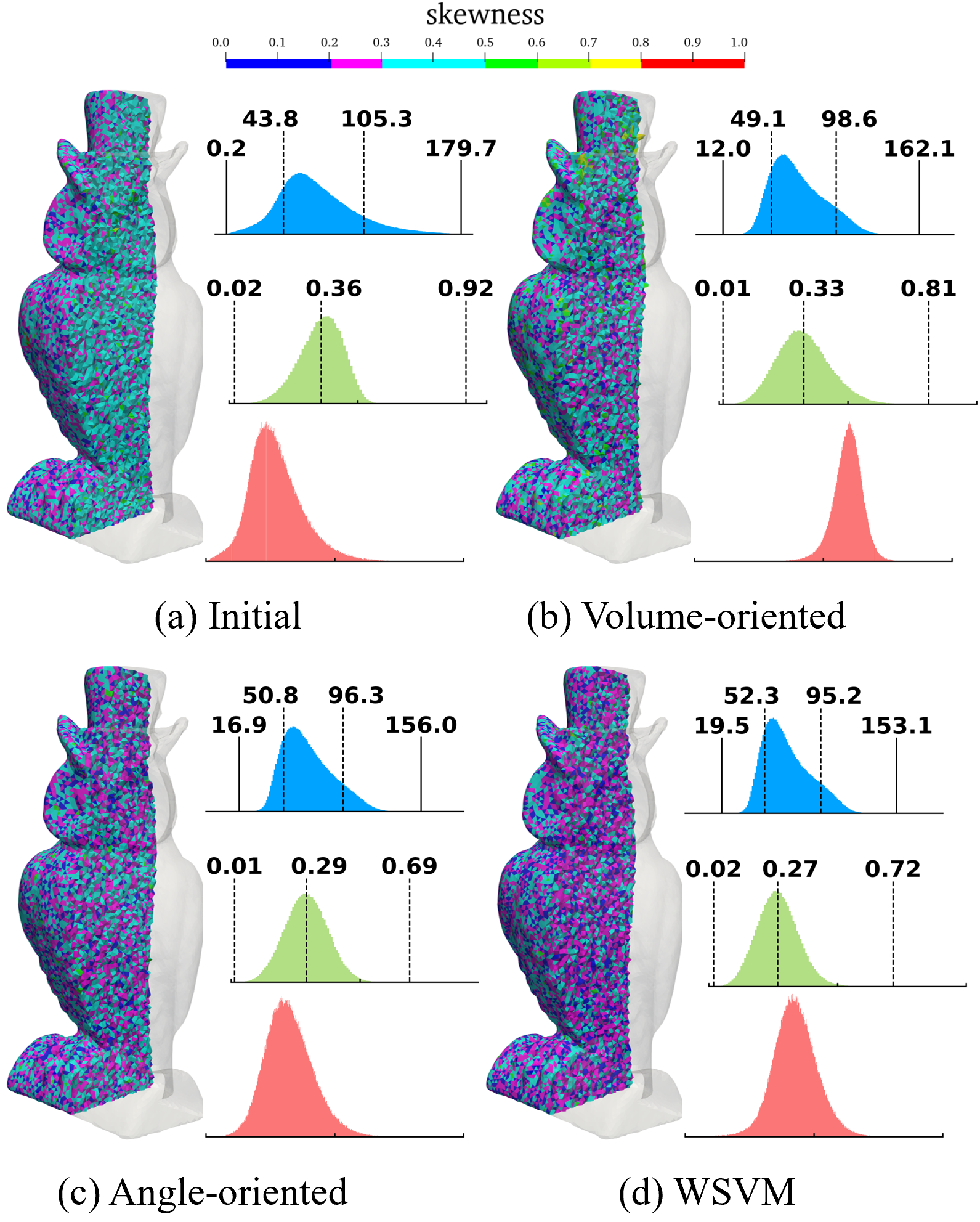}
\caption{Changes in the dihedral angle histogram equiangle skewness histogram and volume histogram of the Owl model at different stages of WSVM.}
\label{fig:owl_skewness}
\end{figure}

\begin{figure}[htbp]
\centering
   \includegraphics[width=3.2in]{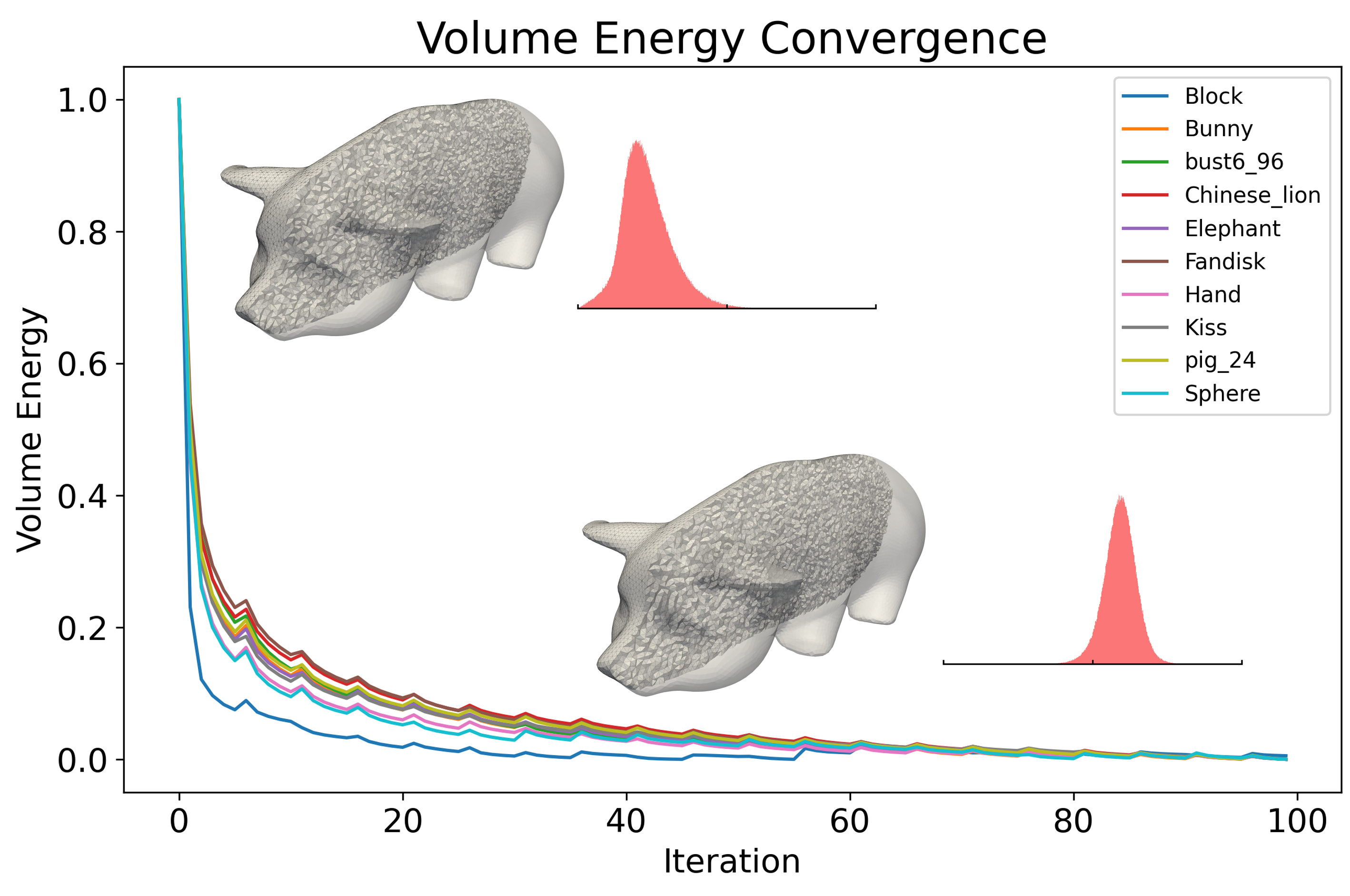}
\caption{Volume-oriented optimization. Ten models were selected to demonstrate the changes in their normalized volume-oriented energy during the optimization stage. A comparison of the volume histograms before and after optimization is shown for one of the models.}\label{fig:volume_energy}
\end{figure}

\begin{figure*}[htbp]
\centering
  \includegraphics[width=3.20in]{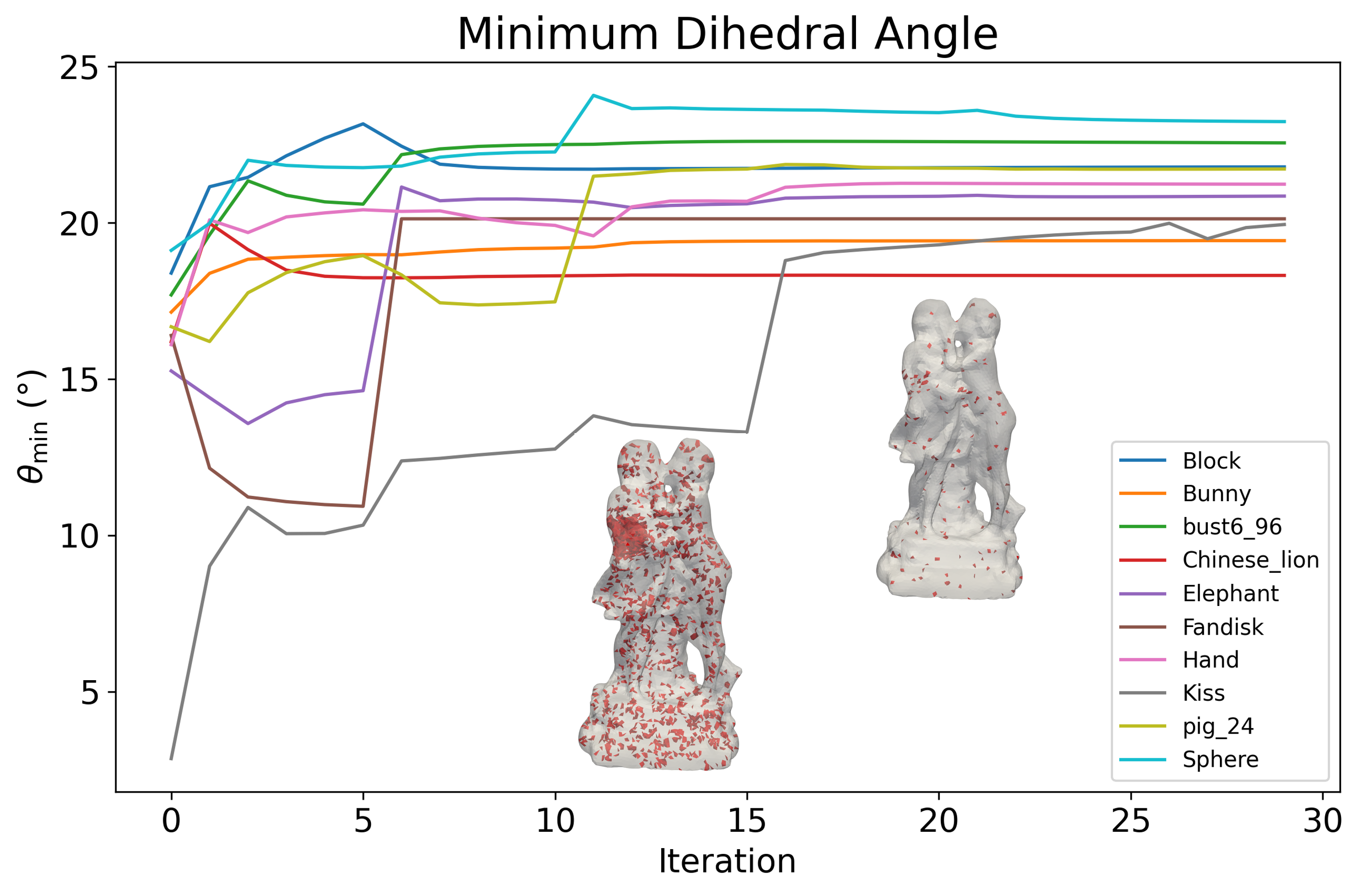}
    \includegraphics[width=3.20in]{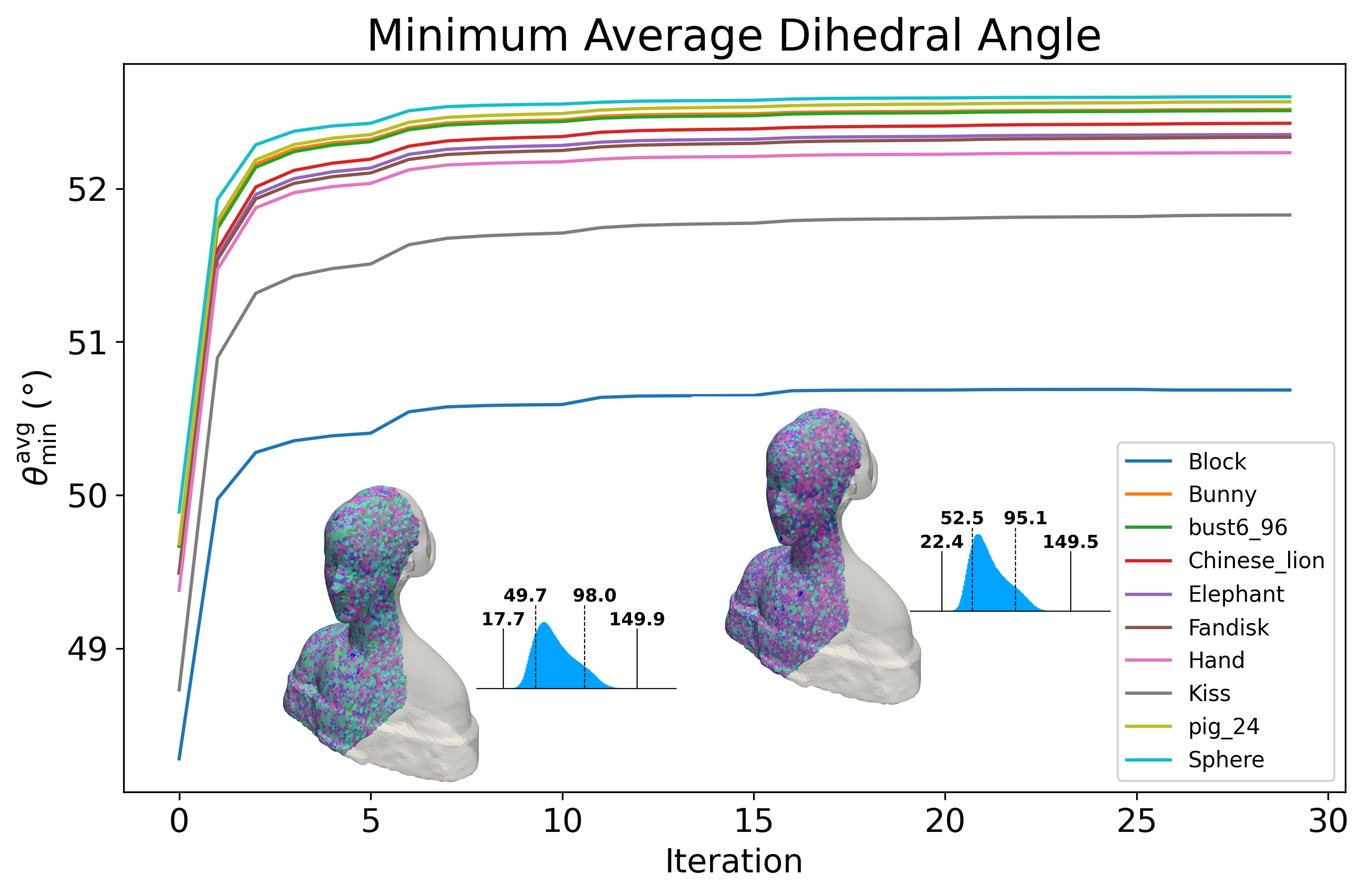}
\caption{Dihedral angle-oriented optimization. Left: The change curves of the minimum dihedral angles for ten models during the angle-oriented optimization process, along with a plot showing the changes in poor-quality elements before and after optimization for one of the models. Right: The average minimum dihedral angles of ten models gradually converge after iterations. A plot of the dihedral angle histogram changes is presented for one of the models.}\label{fig:minAvg_angle_improvement}
\end{figure*}

\paragraph{Ablation studies} 
Given an initial tetrahedral mesh, WSVM enhances its quality by minimizing volume-oriented energy and dihedral angle-oriented energy, respectively. Figure \ref{fig:volume_energy} shows the change curve of the volume-oriented energy during the optimization process. Although there are minor oscillations due to topological operations, as well as split and collapse actions, the overall trend gradually converges. Figure~\ref{fig:minAvg_angle_improvement} illustrates the effect of further optimizing the angles after volume optimization, showing significant improvements in both the minimum angle and the average minimum dihedral angle.

In Figure \ref{fig:owl_skewness}, we use the Owl model to demonstrate the four stages of optimization in WSVM. We observe that while volume optimization alone does not contribute to the mesh quality, it greatly improves the uniformity of the tetrahedral mesh. Angle optimization, on the other hand, significantly improves the mesh quality. By combining these two optimizations, we greatly enhance the overall quality of the generated tetrahedral meshes.


\begin{figure*}[htbp]
\centering
\includegraphics[width=0.95\linewidth]{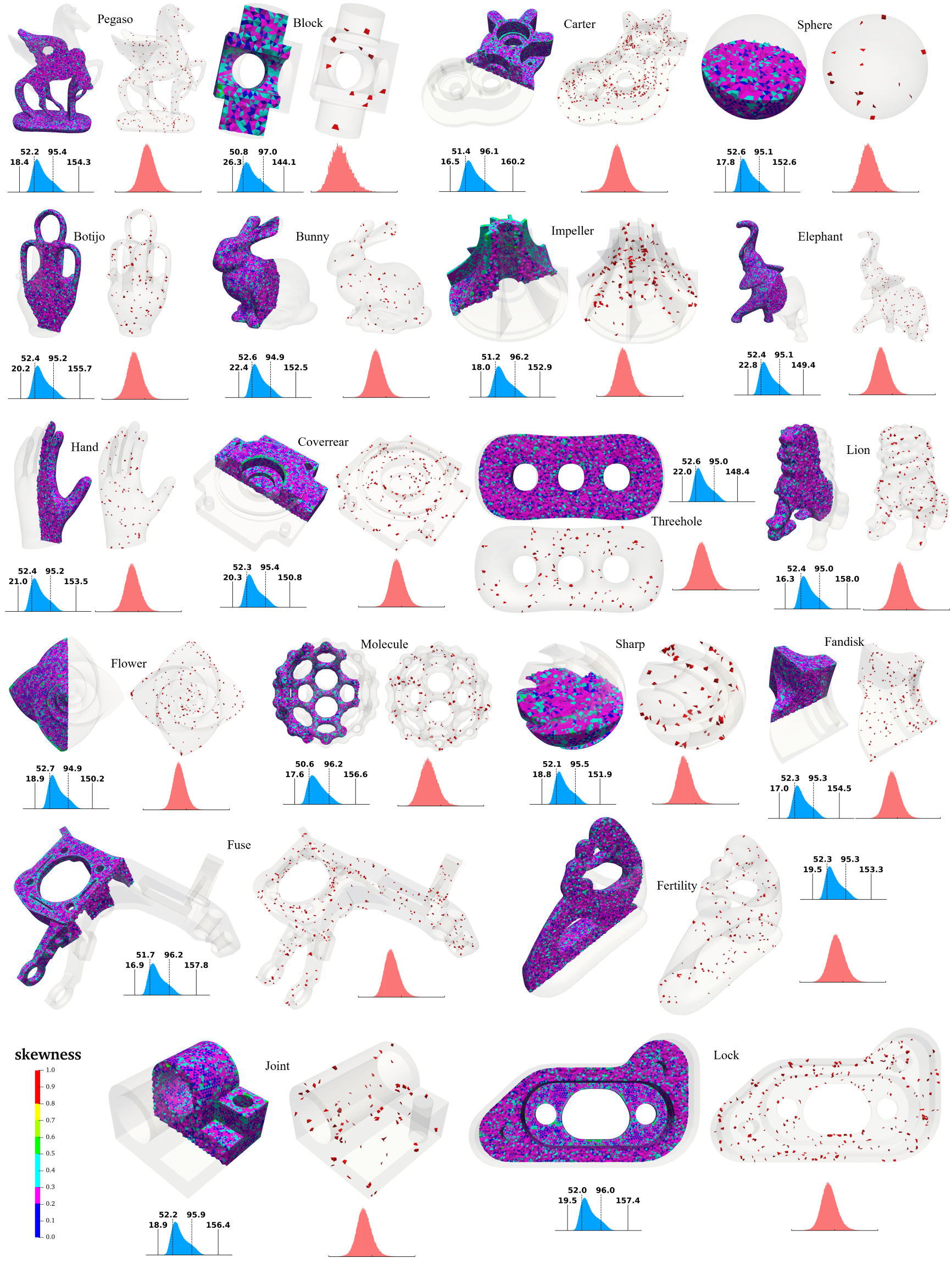}
\caption{Visualization of selected results across the test models. For each model, we present: a slice view of Equiangle Skewness quality, a visualization highlighting tetrahedra with minimum dihedral angles less than 30 degrees, and histograms depicting the distribution of dihedral angles (blue) and tetrahedron volumes (red).
}
\label{fig:gallery}
\end{figure*}

\section{Conclusion and Future Work}\label{sec:conclude}
This paper introduces WSVM, an optimization-based method for generating high-quality, uniform tetrahedral meshes. Specifically, we propose a novel weighted squared volume energy. Minimizing this energy with a constant weight promotes uniform tetrahedral volumes, while an adaptive weight, which considers local geometry, enhances the uniformity of tetrahedral shapes. By integrating these strategies together, we demonstrate that WSVM produces tetrahedral meshes of higher quality--with fewer slivers and greater uniformity of tetrahedral elements--compared to existing methods. 
 
There are several interesting directions for future research. First, parallelizing and enhancing tetrahedral flip operations could significantly improve runtime efficiency. 
Second, exploring alternative density functions with sizing fields could generate non-uniform tetrahedral meshes. Third, extending WSVM to generate higher-order meshes would be valuable for physics-based simulations. Finally, enabling WSVM to accept imperfect inputs would greatly expand its applicability to real-world models.


\bibliographystyle{IEEEtran}
\bibliography{ref} 

{
\appendices \label{Appendix A}

\section{Proof of Theorem \ref{Theorem 1}}\label{Appendix A}

\begin{proof}
Consider the the perturbed function $\mathbf{V}( x,  y,  z) + \varepsilon \delta \mathbf{u}( x,  y,  z)$ of $\mathbf{V}( x,  y,  z)$, where $\varepsilon \in \mathbb{R}$ is the perturbation parameter, and the $\delta \mathbf{u}( x,  y,  z) \in \mathbb{V}$ is the perturbation direction. Since \((\mathbf{V}(x, y, z) + \varepsilon \delta \mathbf{u}(x, y, z)) \in \mathbb{V}\), with the boundary condition
\begin{align*}
&\mathbf{V}(x, y, z)\big|_{\partial \Omega} = \mathbf{S}, \\
&(\mathbf{V}(x, y, z) + \varepsilon \delta \mathbf{u}(x, y, z))\big|_{\partial \Omega} = \mathbf{S},
\end{align*}
the perturbation direction must satisfy
\begin{align*}
\forall \delta \mathbf{u}, \quad \delta \mathbf{u}\big|_{\partial \Omega} = \mathbf{0}.
\end{align*}

The Gâteaux derivative of the functional $E(\mathbf{V})$ with the perturbation $\delta \mathbf{u}$ is $\delta E(\mathbf{V} ; \delta \mathbf{u}) :=\left.\frac{d}{d\varepsilon} \left(\mathbf{E}(\mathbf{V} + \varepsilon \delta \mathbf{u})\right)\right|_{\varepsilon=0}$, we get:

\begin{equation*}
\begin{aligned}
& \delta E(\mathbf{V} ; \delta \mathbf{u}) \\
& =\lim_{\varepsilon \rightarrow 0} \frac{1}{\varepsilon}[\delta E(\mathbf{V}+\varepsilon \delta \mathbf{u})-\delta E(\mathbf{V})] \\
& =\lim_{\varepsilon \rightarrow 0} \frac{1}{\varepsilon} \Bigg\{ 
\iiint_\Omega \rho \Bigg[ (\mathbf{V}_x + \varepsilon \delta \mathbf{u}_x) \cdot \\
& \quad \left( (\mathbf{V}_y + \varepsilon \delta \mathbf{u}_y) \times (\mathbf{V}_z + \varepsilon \delta \mathbf{u}_z) \right) \Bigg]^2 \mathrm{d}x \mathrm{d}y \mathrm{d}z \\
& \quad - \iiint_\Omega \rho \left[ \mathbf{V}_x \cdot (\mathbf{V}_y \times \mathbf{V}_z) \right]^2 \mathrm{d}x \mathrm{d}y \mathrm{d}z \Bigg\} \\
& = \lim_{\varepsilon \rightarrow 0} \frac{1}{\varepsilon} \iiint_\Omega \rho \left( \mathbf{V}_x \cdot (\mathbf{V}_y \times \mathbf{V}_z) \right) \\
& \quad \Bigg[ \varepsilon \mathbf{V}_x \cdot (\mathbf{V}_y \times \delta \mathbf{u}_z) 
+ \varepsilon \mathbf{V}_x \cdot (\delta \mathbf{u}_y \times \mathbf{V}_z) \\
& \quad + \varepsilon \delta \mathbf{u}_x \cdot (\mathbf{V}_y \times \mathbf{V}_z) \Bigg] 
\, \mathrm{d}x \mathrm{d}y \mathrm{d}z.
\end{aligned}
\end{equation*}

After simplification of the above equation, the following result can be obtained:
\begin{equation*}
\begin{aligned}
& \delta E(\mathbf{V} ; \delta \mathbf{u}) \\
&= \iiint_\Omega \left(\rho \mathbf{V}_ x \cdot\left(\mathbf{V}_ y \times \mathbf{V}_ z\right) \left(\mathbf{V}_ x \times \mathbf{V}_ y\right)\right) \cdot \delta \mathbf{u}_ z \, \mathrm{d}x\mathrm{d}y\mathrm{d}z \\
&\quad + \iiint_\Omega \left(\rho \mathbf{V}_ x \cdot\left(\mathbf{V}_ y \times \mathbf{V}_ z\right) \left(\mathbf{V}_ z \times \mathbf{V}_ x\right)\right) \cdot \delta \mathbf{u}_ y \, \mathrm{d}x\mathrm{d}y\mathrm{d}z \\
&\quad + \iiint_\Omega \left(\rho \mathbf{V}_ x \cdot\left(\mathbf{V}_ y \times \mathbf{V}_ z\right) \left(\mathbf{V}_ y \times \mathbf{V}_ z\right)\right) \cdot \delta \mathbf{u}_ x \, \mathrm{d}x\mathrm{d}y\mathrm{d}z.
\end{aligned}
\end{equation*}

In order to simplify the expressions for the variation of energy, we introduce three auxiliary vectors defined as follows:
\begin{equation*}
\begin{aligned}
    \mathbf{I} &= \mathbf{V}_ x \times \mathbf{V}_ y, \\
    \mathbf{J} &= \mathbf{V}_ z \times \mathbf{V}_ x, \\
    \mathbf{K} &= \mathbf{V}_ y \times \mathbf{V}_ z.
\end{aligned}
\end{equation*}

Using the operation rules of vectors, we can obtain the following two simple properties:
\begin{property}\label{Property 1}
\begin{align*}
& \mathbf{I}_z + \mathbf{J}_y + \mathbf{K}_x \\
&= \mathbf{V}_{xz} \times \mathbf{V}_y + \mathbf{V}_x \times \mathbf{V}_{yz} \\
&\quad + \mathbf{V}_{yz} \times \mathbf{V}_x + \mathbf{V}_z \times \mathbf{V}_{xy} \\
&\quad + \mathbf{V}_{xy} \times \mathbf{V}_z + \mathbf{V}_y \times \mathbf{V}_{xz} \\
&= 0
\end{align*}
\end{property}

Based on the relation $\mathbf{V}_x \cdot \mathbf{I} = \mathbf{V}_x \cdot (\mathbf{V}_x \times \mathbf{V}_y) = 0$, we can derive:

\begin{property}\label{Property 2}
\begin{equation*}
\left\{
\begin{array}{l}
\mathbf{V}_x \cdot \mathbf{I} = 0, \\
\mathbf{V}_x \cdot \mathbf{J} = 0, \\
\mathbf{V}_y \cdot \mathbf{I} = 0, \\
\mathbf{V}_y \cdot \mathbf{K} = 0, \\
\mathbf{V}_z \cdot \mathbf{J} = 0, \\
\mathbf{V}_z \cdot \mathbf{K} = 0.
\end{array}
\right.
\end{equation*}
\end{property}

Using the integration by parts formula, we can express the Gâteaux derivative as:

\begin{equation*}
\begin{aligned}
&\delta E(\mathbf{V} ; \delta \mathbf{u}) \\
&= \iiint_\Omega \rho \big(\mathbf{V}_{ x} \cdot \mathbf{K}\big) \mathbf{I} \cdot \delta \mathbf{u}_{ z} \, \mathrm{d}x\mathrm{d}y\mathrm{d}z \\
&+ \iiint_\Omega \rho \big(\mathbf{V}_{ x} \cdot \mathbf{K}\big) \mathbf{J} \cdot \delta \mathbf{u}_{ y} \, \mathrm{d}x\mathrm{d}y\mathrm{d}z  \\
&+ \iiint_\Omega \rho \big(\mathbf{V}_{ x} \cdot \mathbf{K}\big) \mathbf{K} \cdot \delta \mathbf{u}_{ x} \, \mathrm{d}x\mathrm{d}y\mathrm{d}z  \\
&= \rho \big(\mathbf{V}_{ x} \cdot \mathbf{K}\big) \mathbf{I} \cdot \delta \mathbf{u} \big|_{\partial \Omega} 
- \iiint_\Omega \frac{\partial \big(\rho \big(\mathbf{V}_{ x} \cdot \mathbf{K}\big) \mathbf{I}\big)}{\partial  z} \cdot \delta \mathbf{u} \, \mathrm{d}x\mathrm{d}y\mathrm{d}z  \\
&+ \rho \big(\mathbf{V}_{ x} \cdot \mathbf{K}\big) \mathbf{J} \cdot \delta \mathbf{u} \big|_{\partial \Omega} 
- \iiint_\Omega \frac{\partial \big(\rho \big(\mathbf{V}_{ x} \cdot \mathbf{K}\big) \mathbf{J}\big)}{\partial  y} \cdot \delta \mathbf{u} \, \mathrm{d}x\mathrm{d}y\mathrm{d}z  \\
&+ \rho \big(\mathbf{V}_{ x} \cdot \mathbf{K}\big) \mathbf{K} \cdot \delta \mathbf{u} \big|_{\partial \Omega} 
- \iiint_\Omega \frac{\partial \big(\rho \big(\mathbf{V}_{ x} \cdot \mathbf{K}\big) \mathbf{K}\big)}{\partial  x} \cdot \delta \mathbf{u} \, \mathrm{d}x\mathrm{d}y\mathrm{d}z. 
\end{aligned}
\end{equation*}

Since $\forall\delta\mathbf{u},  \left.\delta\mathbf{u}\right|_{\partial\Omega}=\mathbf{0}$, therefore:

\begin{equation*}
\begin{split}
&\delta E(\mathbf{V} ; \delta \mathbf{u}) \\
&= -\iiint_\Omega \left[ 
\frac{\partial \left( \rho \left( \mathbf{V}_{ x} \cdot \mathbf{K} \right) \mathbf{I} \right)}{\partial  z} \right. \\
& \left. + \frac{\partial \left( \rho \left( \mathbf{V}_{ x} \cdot \mathbf{K} \right) \mathbf{J} \right)}{\partial  y} 
+ \frac{\partial \left( \rho \left( \mathbf{V}_{ x} \cdot \mathbf{K} \right) \mathbf{K} \right)}{\partial  x}
\right] \cdot \delta \mathbf{u} \, \mathrm{d}x\mathrm{d}y\mathrm{d}z.\\ 
& =-\iiint_\Omega\nabla E\cdot\delta\mathbf{u}\mathrm{d}x\mathrm{d}y\mathrm{d}z
\end{split}
\end{equation*}

According to the variational principle, a necessary condition for the functional \(E(\mathbf{V})\) to have an extremum at \(\mathbf{V}(x, y, z)\) is that \(\delta E(\mathbf{V} ; \delta \mathbf{u})\) must be zero in every direction \(\delta \mathbf{u}\), therefore, we could obtain $\nabla E =0$. Since:
\begin{equation*}
\begin{aligned}
&\nabla E \\
&=\frac{\partial\left(\rho\left(\mathbf{V}_ x\cdot \mathbf{K}\right)\mathbf{I}\right)}{\partial z}+\frac{\partial\left(\rho\left(\mathbf{V}_ x\cdot \mathbf{K}\right)\mathbf{J}\right)}{\partial y}+\frac{\partial\left(\rho\left(\mathbf{V}_ x\cdot \mathbf{K}\right)\mathbf{K}\right)}{\partial x} \\
&=\frac{\partial\left(\rho\mathbf{V}_ x\cdot \mathbf{K}\right)}{\partial z}\mathbf{I}+\left(\rho\mathbf{V}_ x\cdot \mathbf{K}\right)\mathbf{I}_ z+\frac{\partial\left(\rho\mathbf{V}_ x\cdot \mathbf{K}\right)}{\partial y}\mathbf{J} \\
&+\left(\rho\mathbf{V}_ x\cdot \mathbf{K}\right)\mathbf{J}_ y+\frac{\partial\left(\rho\mathbf{V}_ x\cdot \mathbf{K}\right)}{\partial x}\mathbf{K}+\left(\rho\mathbf{V}_ x\cdot \mathbf{K}\right)\mathbf{K}_ x  \\
&=\frac{\partial\left(\rho\mathbf{V}_{ x}\cdot \mathbf{K}\right)}{\partial z}\mathbf{I}+\frac{\partial\left(\rho\mathbf{V}_{ x}\cdot \mathbf{K}\right)}{\partial y}\mathbf{J}+\frac{\partial\left(\rho\mathbf{V}_{ x}\cdot \mathbf{K}\right)}{\partial x}\mathbf{K} \\
&+\left(\rho\mathbf{V}_ x\cdot \mathbf{K}\right)\left(\mathbf{I}_ z+\mathbf{J}_ y+\mathbf{K}_ x\right)
\end{aligned}
\end{equation*}

With Property \ref{Property 1} we get:
\begin{equation*}
\begin{aligned}
&\nabla E \\
&= \frac{\partial\left(\rho \mathbf{V}_{ x} \cdot \mathbf{K}\right)}{\partial  z} \mathbf{I} 
+ \frac{\partial\left(\rho \mathbf{V}_{ x} \cdot \mathbf{K}\right)}{\partial  y} \mathbf{J} 
+ \frac{\partial\left(\rho \mathbf{V}_{ x} \cdot \mathbf{K}\right)}{\partial  x} \mathbf{K}
\end{aligned}
\end{equation*}


With Property \ref{Property 2}, we could obtain:
\begin{equation*}
\begin{aligned}
& \mathbf{V}_x \cdot \nabla E \\
& = \mathbf{V}_x \cdot \left(\frac{\partial\left(\rho \mathbf{V}_x \cdot \mathbf{K}\right)}{\partial z} \cdot \mathbf{I} + \frac{\partial\left(\rho \mathbf{V}_x \cdot \mathbf{K}\right)}{\partial y} \cdot \mathbf{J} \right.\\ 
& \quad \left. + \frac{\partial\left(\rho \mathbf{V}_x \cdot \mathbf{K}\right)}{\partial x} \cdot \mathbf{K}\right) \\ 
&= \frac{\partial\left(\rho \mathbf{V}_x \cdot \mathbf{K}\right)}{\partial x}\left(\mathbf{V}_x \cdot \mathbf{K}\right)
\end{aligned}
\end{equation*}

With $\nabla E = 0$ and $\mathbf{V}_ x \cdot \nabla E = 0$, we have $\frac{\partial(\rho\mathbf{V}_ x\cdot \mathbf{K})}{\partial  x} (\mathbf{V}_ x \cdot \mathbf{K}) = 0$. However, $\mathbf{V}_ x \cdot \mathbf{K} =\mathbf{V}_x \cdot (\mathbf{V}_y \times \mathbf{V}_z) \neq 0$, which implies $\frac{\partial(\rho\mathbf{V}_ x\cdot \mathbf{K})}{\partial  x} = 0$.

By symmetry, we can obtain similar results for the other partial derivatives: $\frac{\partial(\rho\mathbf{V}_ x\cdot \mathbf{K})}{\partial  y} = 0$ and $\frac{\partial(\rho\mathbf{V}_ x\cdot \mathbf{K})}{\partial  z} = 0$. 

This leads to $\rho(\mathbf{V}_ x\cdot \mathbf{K}) = \rho (\mathbf{V}_ x\cdot(\mathbf{V}_ y\times\mathbf{V}_ z))$ becoming constant.

\end{proof}

\section{Convexity Proof}\label{appd:convexProof}

\begin{lemma}\label{equ:lemma2_subj}
If the intersection of all half-spaces \(\{ f_m(\boldsymbol{v}_i) > 0 \}\) for \( m = 1 \ldots n \) is non-empty, then it forms a convex set.
\end{lemma}

\begin{proof}[Proof of Lemma \ref{equ:lemma2_subj}]
Each constraint \( f_m(\boldsymbol{v}_i) > 0 \) defines a half-space, which is a convex set. By the properties of convex sets, the intersection of any collection of convex sets is also convex. Therefore, if the intersection of the half-spaces \(\bigcap_{m=1}^n \{ \boldsymbol{v}_i \in \mathbb{R}^3 \mid f_m(\boldsymbol{v}_i) > 0 \}\) is non-empty, it is inherently convex. 
\end{proof}

\begin{lemma}\label{equ:lemma1_obj}
The objective function \(E(\boldsymbol{v}_i)\) is convex with respect to \(\boldsymbol{v}_i\).
\end{lemma}

\begin{proof}[Proof of Lemma \ref{equ:lemma1_obj}]
A function is convex if its Hessian matrix is positive semi-definite.

To find the Hessian of \(E(\boldsymbol{v}_i)\), we derive the gradient:


\begin{equation*}
\begin{aligned}
\boldsymbol{g}
&=\nabla E(\boldsymbol{v}_{i}) \\
&= \sum_{m=1}^{n} 2 \rho_m \overline{V}_m \frac{\partial \overline{V}_m}{\partial \boldsymbol{v}_{i}} \\
&= \sum_{m=1}^{n} 2 \rho_m \overline{V}_m \frac{\partial}{\partial \boldsymbol{v}_{i}} \left[ \frac{1}{6}  (\boldsymbol{v}_{mk} - \boldsymbol{v}_{ml}) \cdot \left( (\boldsymbol{v}_i - \boldsymbol{v}_{ml}) \right. \right. \\
& \hspace{8em} \left. \left. \times (\boldsymbol{v}_{mj} - \boldsymbol{v}_{ml}) \right) \right] \\
&= \sum_{m=1}^{n} \frac{1}{3} \rho_m \overline{V}_m \left( (\boldsymbol{v}_{mj} - \boldsymbol{v}_{ml}) \times (\boldsymbol{v}_{mk} - \boldsymbol{v}_{ml}) \right).
\end{aligned}
\end{equation*}

Denoting that \(\boldsymbol{s_m} = (\boldsymbol{v}_{mj} - \boldsymbol{v}_{ml}) \times (\boldsymbol{v}_{mk} - \boldsymbol{v}_{ml})\), the Hessian of \(E(\boldsymbol{v}_i)\) is derived as:

\begin{equation*}
\begin{aligned}
\boldsymbol{H}
&= \nabla^2 E(\boldsymbol{v}_{i})\\
&= 2 \sum_{m=1}^{n} \rho_m \left(\frac{\partial}{\partial \boldsymbol{v}_{i}} \left( \overline{V}_m \boldsymbol{s}_m \right)\right)\\
&= 2 \sum_{m=1}^{n} \rho_m \left( \frac{\partial \overline{V}_m}{\partial \boldsymbol{v}_{i}} \otimes \boldsymbol{s}_m + \overline{V}_m \frac{\partial \boldsymbol{s}_m}{\partial \boldsymbol{v}_{i}} \right)\\
&= \frac{1}{18} \sum_{m=1}^{n} \rho_m \boldsymbol{s}_m \otimes \boldsymbol{s}_m.
\end{aligned}
\end{equation*}

For any vector \(\boldsymbol{x}\), the following condition holds:

\[
\boldsymbol{x}^\top \boldsymbol{H} \boldsymbol{x} = 2 \sum_{m=1}^{n} \rho_m (\boldsymbol{x}^\top \boldsymbol{s_m})^2,
\]
which is always non-negative. Therefore, the Hessian is positive semi-definite, indicating that \(E(\boldsymbol{v}_i)\) is convex.
\end{proof}

\begin{proof}
Lemma \ref{equ:lemma2_subj} verifies that each constraint \(f_m(\boldsymbol{v}_i) > 0\) forms a convex set.
Concurrently, Lemma \ref{equ:lemma1_obj} establishes the convexity of the objective function \(E(\boldsymbol{v}_i)\), shown by its positive semi-definite Hessian.
Therefore, with a convex objective and convex constraints, the optimization problem satisfies the criteria for convexity.

\end{proof}
}



\newpage
 
\vspace{11pt}

\vspace{11pt}

\vfill

\end{document}